\documentclass{article}[13pt]

\usepackage{amsmath,amssymb,amsthm,amsfonts,mathtools}
\usepackage{fullpage}
\usepackage{graphicx,caption,subcaption}
\usepackage{booktabs,makecell,array}
\usepackage{enumerate}
\usepackage{dsfont}
\usepackage{color}
\usepackage{lmodern}
\usepackage[T1]{fontenc}
\usepackage[utf8]{inputenc}
\usepackage[unicode=true]{hyperref}
 \hypersetup{
   pdfstartview={FitH},
   pdffitwindow=false,
   pdfborder={0 0 0},
   colorlinks=true,
   linktoc=page,
   linkcolor=blue,
   urlcolor=blue,
   hypertexnames=false,
   pdfdisplaydoctitle=true,
   pdfduplex=DuplexFlipLongEdge,
 }

\numberwithin{equation}{section}

\newtheorem{theoreme}{Theorem}[section]

\newtheorem{lemma}[theoreme]{Lemma}

\theoremstyle{definition}
\newtheorem{remark}[theoreme]{Remark}

\newcommand{\suma}[2]{\sum\limits_{#1}^{#2}}

\newcounter{smallarabics}
\newenvironment{arabicenumerate}
{\begin{list}{{\normalfont\textrm{(\arabic{smallarabics})}}}
  {\usecounter{smallarabics}\setlength{\itemindent}{0cm}
   \setlength{\leftmargin}{5ex}\setlength{\labelwidth}{4ex}
   \setlength{\topsep}{0.75\parsep}\setlength{\partopsep}{0ex}
   \setlength{\itemsep}{0ex}}}
{\end{list}}

\newcounter{smallroman}

\newcommand{\ben}{\begin{arabicenumerate}}
\newcommand{\een}{\end{arabicenumerate}}

\def\rr{{\mathbb R}}

\def\zz{{\mathbb Z}}
\def\cc{{\mathbb C}}
\def\nn{{\mathbb N}}

\def\ii{{\rm i}}

\def\oplusal{\mathop{\hbox{\raise 1.5 ex
  \hbox{$\scriptscriptstyle\rm al$}
\kern -0.92 em \hbox{$\oplus$}}}}
\def\otimesal{\mathop{\hbox{\raise 1.5 ex
  \hbox{$\scriptscriptstyle\rm al$}
\kern -0.92 em \hbox{$\otimes$}}}}
\def\Gammal{\hbox{\raise 1.68 ex 
\hbox{$\scriptscriptstyle\rm al$}\kern -0.50 em $\Gamma$}}

\def\bar{\overline}

\def\loc{{\rm loc}}

\renewcommand\Re{{\rm Re}}

\def\i{{\rm i}}

\def\sgn{{\rm sgn}}

\def\Tr{{\rm Tr}}

\def\e{{\rm e}}
\def\d{{\rm d}}

\def\cD{{\cal D}}

\def\cH{{\cal  H}}

\def\cJ{{\cal J}}

\def\cI{{\cal I}}
\def\cN{{\cal N}}

\def\cK{{\cal K}}
\def\cB{{\cal B}}

\def\cW{{\cal W}}
\def\cU{{\cal U}}

\def\Im{{\rm Im}}

\def\wlim{{\rm w-}\lim}

\def\12{\tfrac{1}{2}}
\def\32{\tfrac{3}{2}}
\def\52{\tfrac{5}{2}}

\def\p{ \partial}

\def\rN{\mathrm{N}}
\def\rD{\mathrm{D}}

\def\isoI{\mathbb{I}}
\def\isoK{\mathbb{K}}
\usepackage[normalem]{ulem} 

\begin{document}
\title{Exactly  solvable Schr\"odinger operators \\ related to the confluent equation}

\author{
	Jan Derezi\'{n}ski\\
	\small Department of Mathematical Methods in Physics,\\
	\small Faculty of Physics, University of Warsaw, \\
	\small Pasteura 5, 02-093 Warszawa, Poland\\
	\small email: \texttt{jan.derezinski@fuw.edu.pl}\\
	\and
	Jinyeop Lee\\
	\small Department of Mathematics and Computer Science,\\
	\small University of Basel,\\
	\small Spiegelgasse 1, CH-4051, Basel, Switzerland\\
	\small email: \texttt{jinyeop.lee@unibas.ch}\\
}

\maketitle

\begin{abstract}
Our paper investigates one-dimensional Schr\"odinger operators defined as closed operators on $L^2(\rr)$ or $L^2(\rr_+)$ that are exactly solvable in terms of confluent functions (or, equivalently, Whittaker functions). We allow the potentials to be complex. They fall into three families:  Whittaker operators (or radial Coulomb Hamiltonians), Schr\"odinger operators with Morse potentials and isotonic oscillators. For each of them, we discuss  the corresponding basic holomorphic family of closed operators and the integral kernel of their resolvents. We also describe transmutation identities that relate these resolvents. These identities interchange spectral parameters with coupling constants across different operator families.

A similar analysis is performed for one-dimensional Schr\"odinger operators solvable in terms of Bessel functions (which are reducible to  special cases of Whittaker functions). They fall into  two families: Bessel operators and Schr\"odinger operators with exponential potentials.

To make our presentation self-contained, we include a short summary of  the theory of closed one-dimensional Schr\"odinger operators with singular boundary conditions. We also provide a concise review of special functions that we use.
\end{abstract}

\section{Introduction}

{\em One-dimensional Schr\"odinger operators} are operators of the form
\begin{equation}
	\label{oper} L := -\partial_x^2 + V(x),
\end{equation}
where $V(x)$ is the {\em potential}, which in this paper is allowed to be a complex-valued function.
In some rare cases,  (not necessarily square integrable) eigenfunctions of \eqref{oper} can be  computed in terms of standard special functions. We call such operators {\em exactly solvable}.

Our paper is devoted to several families of operators of the form \eqref{oper}, interpreted as {\em closed operators} on $L^2(]a,b[)$ for appropriate $-\infty\leq a<b\leq+\infty$, which are exactly solvable in terms of confluent functions, or equivalently, Whittaker functions. 

We study three categories of one-dimensional Schr\"odinger operators:
\begin{enumerate}
	\item[(1)] those reducible to the Bessel equation, having two families: the Bessel operator and the Schr\"odinger operator with an exponential potential \cite{DeAlfaroRegge,DW1,DunfordSchwartz,EdmundsEvans,Kato,Pankrashkin};
	\item[(2)] those associated with the Whittaker equation, which includes three families: the Whittaker operator, the isotonic oscillator, and the Schr\"odinger operator with Morse potential \cite{DW1,Morse}; 
	\item[(3)] finally, for completeness, we include the (well-known) harmonic oscillator, reducible to the Weber equation.
\end{enumerate}

Note that Bessel and Weber functions can be reduced to subclasses of Whittaker functions. 
One can also remark that there exist several classes of one-dimensional Schr\"odinger equations exactly solved in terms of the Gauss hypergeometric function \cite{bose,DW1,GPS,whittaker1962course}, which are not considered in this paper.

On the algebraic level, all these families are discussed in many references, e.g., \cite{Chuan1991,Cooper1987,DW1,Dutt1988}. 
They are usually treated as formal differential expressions without a functional-analytic setting. 
In this paper, we consider them as closed operators on an appropriate Hilbert space. 
For exactly solvable Schr\"odinger operators, we are able to express their resolvent in terms of special functions. 
We can do the same with eigenprojections and the spectral measure. 
Sometimes we can also compute other related operators, such as their exponential or Møller (wave) operators. 

Exactly solvable Schr\"odinger operators, interpreted as closed (usually self-adjoint) operators, are essential in applications, serving as reference models for various perturbation and scattering problems. 
They are widely used by theoretical physicists to model real quantum systems and for instructional purposes in quantum mechanics (see, e.g., \cite{Flugge,GK}).

	{\em Sturm--Liouville operators}, which are given by an expression of the form
	\begin{equation}
		-\frac{1}{w(x)}\partial_x p(x)\partial_x + \frac{q(x)}{w(x)},
	\end{equation}
can be reduced to one-dimensional Schr\"odinger operators by the so-called  {\em Liouville transformation} (at least for real $\frac{p(x)}{w(x)}$) \cite{Ency,Liouville} (see also \cite{DW1} Subsections 2.1 and 2.3).

Therefore, our paper is also related to many works describing Sturm--Liouville operators,   such as \cite{Everitt,GeZin,GTV,ReedSimon,Sch}, which, however, are usually restricted to real potentials.

We always choose the Hilbert space to be $L^2(]a,b[)$, where $a$ and $b$ are singular points of the corresponding eigenvalue equation, with the possibility that $a = -\infty$  and $b = +\infty$. With these endpoints, for most parameters, in order to define a closed realization of \eqref{oper}, there is no need to choose  boundary conditions (b.c.). There are, however, exceptional parameter ranges where b.c.\ must be selected. 
Following, e.g., \cite{DeGe}, we will say that the endpoint has index $0$ if b.c.\ are not needed. We will say that it has index $2$ otherwise.

\begin{remark}
	Throughout the paper, we use terminology appropriate for complex potentials.
	Most readers are probably more familiar with related concepts applicable to {\em real} potentials,
	where one is usually interested in finding {\em self-adjoint} realizations of \eqref{oper}.
	For real potentials, the case of index $2$ is often called ``limit circle,'' and the case of 
	index $0$ ``limit point,'' which is not quite appropriate for real potentials (see Remark \ref{limitcircle}).
	
	Note also that for real potentials the operator \eqref{oper} restricted to 
	$C_\mathrm{c}^\infty(]a,b[)$ is {\em Hermitian} (or, as it is commonly termed, {\em symmetric}). 
	The three cases — both endpoints have index 0, one of them has index $0$ and the other index $2$, 
	and both have index $2$ — correspond to the deficiency indices $(0,0)$, $(1,1)$, and $(2,2)$, respectively.
	However, if the potential is not real, then this operator is not Hermitian, and deficiency indices 
	are, in principle, not well-defined.
\end{remark}

The operators that we study depend on (complex) parameters. We try to organize them into {\em holomorphic families of closed operators}, as advocated, e.g., in \cite{Derezinski2, DW2}. To find such families, first, we identify a subset of parameters that uniquely determine a closed extension. In the terminology we use, $L$ has index zero at both endpoints. In all cases that we consider, this subset forms a large set of parameters with a non-empty interior. We obtain a holomorphic family of closed realizations of $L$, which is then extended to its largest possible domain of holomorphy. The holomorphic family obtained in this way will be called {\em basic}. In Table \ref{tab:all-operators}, we describe all basic holomorphic families of Schr\"odinger operators considered in this paper.

\begin{table}[ht]
	\caption{Families of operators  covered in this paper}
	\resizebox{\textwidth}{!}{
		\renewcommand{\arraystretch}{1.5}
		\begin{tabular}{@{}llrll@{}}
			\toprule
			\textbf{Operator or Potential} & \textbf{Symbol and Formula} & \textbf{\makecell[l]{Spectral\\ Parameter}} & \textbf{\makecell[l]{Domain of\\ Holomorphy}} & $L^2(]a,b[)$ \\
			\midrule
			\textbf{\textit{Bessel equation:}} & & & & \\
			\addlinespace
			Bessel operator & 
			$H_m := -\partial_r^2 + \left( m^2 - \frac{1}{4} \right) \frac{1}{r^2}$ & 
			$k^2$ \phantom{  }& $\Re(m) > -1$ & $L^2(\mathbb{R}_+)$ \\
			\addlinespace
			Exponential potential & 
			$M_k := -\partial_x^2 + k^2 \e^{2x}$ & 
			$m^2$ \phantom{  }& $\Re(k) > 0$
			& $L^2(\mathbb{R})$ \\
			\midrule
			\textbf{\textit{Whittaker equation:}} & & & & \\
			\addlinespace
			Whittaker operator & 
			$H_{\beta,m} := -\partial_r^2 + \left( m^2 - \tfrac{1}{4} \right) \frac{1}{r^2} - \frac{\beta}{r}$ & 
			$k^2$ \phantom{  }& \makecell[l]{$\Re(m) > -1$,\\ $(\beta, m) \neq (0, -\tfrac{1}{2})$} & $L^2(\mathbb{R}_+)$ \\
			\addlinespace
			Morse potential & 
			$M_{\beta,k} := -\partial_x^2 + k^2 \e^{2x} - \beta \e^x$ & 
			$m^2$ \phantom{  }& \makecell[l]{$\Re(k) > 0$ 
				\\ $\beta \in \cc\setminus \rr_{+}$} & $L^2(\mathbb{R})$ \\
			\addlinespace
			Isotonic oscillator & 
			$N_{k,m} := -\partial_v^2 + \left(m^2-\tfrac{1}{4}\right)\frac{1}{v^2} + k^2 v^2$ & 
			$-2\beta$ \phantom{  }&\makecell[l]{ $\Re(m) > -1$,\\ $\Re(k)
				> 0$, 
			} & $L^2(\mathbb{R}_+)$ \\
			\midrule
			\textbf{\textit{Weber equation:}} & & & & \\
			\addlinespace
			Harmonic oscillator & 
			$N_k := -\partial_r^2 +k^2 r^2$ & $-2\beta$ \phantom{  }&
			$\Re(k) > 0$
			& $L^2(\mathbb{R})$ \\
			\bottomrule
		\end{tabular}
	}
	\label{tab:all-operators}
\end{table}

Schr\"odinger operators with potentials from Table \ref{tab:all-operators} can have realizations that do not belong to their basic holomorphic families. This occurs, in particular, when the
index of an endpoint is 2 for a given parameter. 
In such cases, there is a whole family of closed realizations of a single formal expression, of which at most two are basic. 
 Boundary conditions that define a closed realization that does not belong to a basic holomorphic family will be called {\em mixed b.c.}
In particular, Bessel and Whittaker operators, as well as the isotonic oscillator, can have mixed b.c.\ at $0$ for $-1<\Re(m)<1$. 
They are not discussed in this paper. For further details, see \cite{Derezinski1} for Bessel operators and \cite{DFNR} for Whittaker operators.

The domain of holomorphy, by definition, is an open set. Schr\"odinger operators with exponential and Morse potentials, as well as the harmonic and isotonic potentials, have a domain of holomorphy with a boundary at $\Re(k)=0$. The basic holomorphic family can be extended by continuity to this boundary. In all cases, operators with parameters on this boundary must be treated separately. We discuss them all, with the exception of Morse potentials.

\medskip

If $A$ is an operator on $L^2(]a,b[)$, then the integral (or possibly distributional) kernel of the operator $A$ will be denoted $A(x,y), \; x,y\in\,]a,b[$. I.e.,
\begin{equation} 
	(f|Ag) = \int \bar{f(x)}A(x,y)g(y) \, \d x \d y.
\end{equation}
As mentioned above, for all operators from Table
\ref{tab:all-operators}, we will compute their spectrum $\sigma$, their point spectrum $\sigma_\mathrm{p}$, and for parameters outside their spectrum the integral
kernel of their resolvent. 
(By the point spectrum we mean the set of eigenvalues with square integrable eigenfunctions).
We will see that Hamiltonians with very different properties can be solved in terms of the same class of special functions. 
It is curious to observe that they are linked by somewhat mysterious identities involving their resolvents.
In these identities, we can observe a ``transmutation'' of spectral parameters  where the resolvents are evaluated into coupling constants.

Here is the transmutation identity involving the resolvent of the Bessel operator and of the Schr\"odinger operator with the exponential potential:
\begin{equation}\label{regge}
(M_k+m^2)^{-1}(x,y)
	=\e^{-\frac{x}{2}}(H_m+k^2)^{-1}(\e^x,\e^y)\,\e^{-\frac{y}{2}}.
\end{equation}
Note that on the left
$m^2$ is the spectral parameter, and on the right it is the coupling constant. On the left   $k^2$ is the coupling constant, on the right it is the spectral parameter. Thus, their roles are curiously interchanged.

For the three families solved in terms of Whittaker functions, we have transmutations involving three parameters: $m$, $k$, and $\beta$:
\begin{align}
	(N_{k,m}-2\beta)^{-1}(u,v) 
	&=2 \Big(\frac{u^2}2\Big)^{-\frac14} (H_{\beta,\frac{m}{2}}+k^2)^{-1}\Big(\frac{u^2}{2},\frac{v^2}{2}\Big)
	\Big(\frac{v^2}{2}\Big)^{-\frac14},\\
	(M_{\beta,k}+m^2)^{-1}(x,y)
	&=\e^{-\frac{x}{2}}
(H_{\beta,m}+k^2)^{-1}(\e^x,\e^y)
	\e^{-\frac{y}{2}},\\
	(N_{k,m}-2\beta)^{-1}(u,v)
	&=2 \Big(\frac{u^2}2\Big)^{\frac14} 
\Big(M_{\beta,k} + \left(\frac{m}{2}\right)^2\Big)^{-1}
	\Big(\log\frac{u^2}{2},\log\frac{v^2}{2}\Big)
	\Big(\frac{v^2}2\Big)^{\frac14}.
\end{align}

The case $\Re(k)=0$ of the exponential potential is quite curious
and analyzed in our paper. Setting $k=\i\ell$, we can rewrite it as:
\begin{equation}  
	M^\gamma_{\i\ell}:=-\partial_x^2-\ell^2\e^{2x}.
\end{equation}
In this operator, we need to fix a b.c.\ at $\infty$, which can be naturally parametrized by $\gamma\in\cc\cup\{\infty\}$.
The parameter $\ell > 0$ is not very interesting --- in fact, the translation by $\ln\ell$ yields unitary
equivalence with the case $\ell = 1$. 
Note that $\gamma\mapsto M^\gamma_{\i\ell}$ is a family of operators holomorphic away from $0$ or $\infty$. 
We prove that $M^\gamma_{\i\ell}$ with $\gamma = 0,\infty$ can be reached as limits of the family $M_k$.
The transmutation of the resolvents for $M^\gamma_{\i\ell}$ into resolvents of $H_m$ is quite interesting:
\begin{equation}\label{regge1}
	\begin{split}
	(M^\gamma_{\i\ell}+m^2)^{-1}(x,y)
		&=\frac{\e^{\i m\pi}}{(\e^{\i m\pi}-\gamma)} \e^{-\frac{x}{2}}{\ell}{(-\ell^2+\i0+H_m)^{-1}}( \e^x, \e^y)\e^{-\frac{y}{2}}\\
		&\quad+\frac{\gamma}{(\e^{\i m\pi}-\gamma)}
		\e^{-\frac{x}{2}}\ell(-\ell^2-\i0+H_m)^{-1}( \e^x,
		\e^y)\e^{-\frac{y}{2}}.
	\end{split}
\end{equation}
Note also that $M_{\i\ell}^\gamma$ for $\gamma\neq0,\infty$ possess point
spectrum at $\e^{\i m\pi}=\gamma$. This spectrum was recently described independently in \cite{Stempak}.

Let us briefly describe the structure of the paper.
We begin with a concise overview of the basic theory of 1-dimensional Schr\"odinger operators, presented in Section \ref{sec:BasicTheory-1D-SchOp}. 
This is a classic subject covered in various textbooks. 
We primarily follow the presentation in \cite{DeGe}. 
We explain how to determine when boundary conditions are needed to define a closed realization with a non-empty resolvent set. 
We also demonstrate how to fix b.c.\ with the help of the Wronskian. 
Finally, we describe the integral kernel for a candidate of the resolvent. 

In Section \ref{sec:Bessel}, we discuss families of closed operators solved in terms of Bessel functions; 
in Section \ref{sec:whittaker}, the families solved in terms of Whittaker functions; and in Section \ref{sec:harmonic}, the harmonic oscillator, solved in terms of Weber functions.

In the appendices, we concisely describe elements of the theory of special functions used in our paper.
In Appendix \ref{Bessel equation}, we discuss various kinds of Bessel functions, as well as the ${}_0F_1$ functions.
First of all, one can distinguish between the trigonometric and hyperbolic Bessel equations. The former has oscillating solutions, while the solutions of the latter behave exponentially. 
One can pass from the former to the latter by rotating the complex plane by the angle $\pm\frac{\pi}{2}$. 
Secondly, both the trigonometric and hyperbolic Bessel equations have forms adapted to all dimensions.
In most of the literature, the standard Bessel equation, which is 2-dimensional and trigonometric, and the modified Bessel equation, which is 2-dimensional and hyperbolic, are considered.
In our paper, it is also convenient to use functions that solve the 1-dimensional Bessel equation, both trigonometric and hyperbolic.
All these variations of the Bessel equation are equivalent to one another, and they are also equivalent to the so-called ${}_0F_1$ equation.
They can also be reduced to a subclass of the confluent equation.

The Whittaker equation is equivalent to the confluent equation, which exists in two equivalent variants: ${}_1F_1$ and ${}_2F_0$. 
Therefore, the {\em confluent equation}, which we put in the title of this manuscript, can be replaced by the {\em Whittaker  equation}. 
There are also varieties of the Whittaker equation for any dimension. 
In our paper, we use those corresponding to $d=1$ and $d=2$. 
We also introduce modifications of Whittaker functions adapted to the isotonic oscillator. All of this is described in Appendix \ref{Whittaker equation}.

Finally, in Appendix \ref{sec:weber}, we briefly describe the Weber equation, which is equivalent to the Hermite equation and to a subclass of the Whittaker equation.

In order to define closed realizations of \eqref{oper}, we will specify their {\em operator domain} and compute their resolvents. 
There exists a different strategy for dealing with unbounded operators: one can try to specify their {\em form domain}. Strictly speaking, the latter strategy, in its orthodox version,
is limited to certain classes of operators, such as positive operators that are bounded from below \cite{ReedSimon,Sch}.

In our paper, we do not discuss the form domains of the operators under study.
We expect that this topic will be addressed in a separate paper, generalizing the analysis of Bessel operators defined as bilinear forms in \cite{DeGe2}.

\section{Basic theory of 1-dimensional Schr\"odinger operators}\label{sec:BasicTheory-1D-SchOp}

\subsection{Boundary conditions}

Let us sketch the theory of closed realizations of 1-dimensional Schr\"odinger operators with complex potentials. 
This is a classic subject, discussed in various textbooks and presented in several forms.  
Our presentation follows mostly \cite{DeGe}. 
Note that we avoid using the so-called boundary triplets and Krein-type formulas.

Consider an open interval $]a,b[$, where $-\infty \leq a < b \leq +\infty$.
(In our paper, we will consider only $\rr$ and $\rr_+$; 
however, the theory in this section applies to an arbitrary interval $]a,b[$). 
Consider a complex function $V \in L_\loc^1(]a,b[)$.
Suppose $L$ is formally given by
\begin{equation}
	\label{oper1} L:=-\partial_x^2+V(x).
\end{equation}
We would like to describe  closed realizations of $L$ on $L^2(]a,b[)$ 
possessing a non-empty resolvent set and compute its resolvent.
(A linear operator possessing a non-empty resolvent set is sometimes called {\em well-posed} \cite{EdmundsEvans,DeGe}).

Note that in most of the literature only  real $V$ are considered, and the authors are interested only in self-adjoint realizations of $L$. 
If $L_\bullet$ is such a realization of $L$, its spectrum, denoted $\sigma(L_\bullet)$, is contained in $\rr$. 
Thus its resolvent set is automatically non-empty.
In our paper, however, we will consider also complex potentials and non-self-adjoint  realizations.

Let $z\in\cc$. Let $AC^1(]a,b[)$ denote the set of functions on the open interval $]a,b[$ whose derivative is absolutely continuous. 
The space 
\begin{equation}
	\cN(L-z):=\big\{\Psi\in AC^1(]a,b[)\ | \ (L-z)\Psi(x)=0 \big\}
\end{equation}
is 2-dimensional. Let $\cU_a(z)$ denote the subspace of $\cN(L-z)$ solutions square integrable near $a$.
One can show that 
one of the following holds:
\begin{align}
	\text{ either }\dim\cU_a(z)=2\text{ for all }z\in\cc;\\
	\text{ or }\dim\cU_a(z)\leq 1\text{ for all }z\in\cc.
\end{align}
In the first case we say that the index of the endpoint $a$, denoted $\nu_a(L)$, is $2$, and in the second case it is $0$.
There are analogous definitions associated to the second endpoint $b$.

\begin{remark}  \label{limitcircle}
	In the context of real potentials, the case $ \nu_a(L) = 2 $ is called “limit circle” and $ \nu_a(L) = 0 $ is called “limit point”. 
	 This terminology goes back to an old paper by H. Weyl \cite{Weyl} and is explained, e.g., in \cite[Appendix to X.1.]{ReedSimon}. 
		These names are not fully appropriate for complex potentials. 
		For example, if $\Im V < 0$, then an analysis similar to Weyl's yields a trichotomy instead of a dichotomy \cite{Sims}; see also \cite[Subsection 8.5]{DeGe}.
\end{remark}

Let us define
\begin{align}\label{piuo}
	\cD(L^{\max})
	& := \big \{ f \in L^2(]a,b[) \, \cap \, AC^1(]0,\infty[) \; \mid\; L f \in L^2(]a,b[)\; \big \} ,\\
	\cD(L^{\min})
	&:=\text{the closure in the graph norm of }\{f\in\cD(L^{\max})\; \mid\; f=0\text{ near }a\text{ and }b\}.
\end{align}
(The graph norm is $\|f\|_A:=\sqrt{\|Lf\|^2+\|f\|^2}$.)
One can show that
\begin{equation}
	\dim\cD(L^{\max})/\cD(L^{\min})=\nu_a(L)+\nu_b(L).\label{endpo}
\end{equation}
We define the maximal and minimal realization of $L$:
\begin{equation}
	L^{\max}:=L\Big|_{\cD(L^{\max})},\quad
	L^{\min}:=L\Big|_{\cD(L^{\min})}.
\end{equation}

In what follows, we will need the Wronskian for two functions $\Phi$ and $\Xi$, defined as 
\begin{equation}\label{def:wronskian}
	\cW(\Phi,\Xi)(x):=\Phi(x)\Xi'(x)-\Phi'(x)\Xi(x). 
\end{equation}
Note that if $\Phi,\Xi\in\cN(L-z)$, then $\cW(\Phi,\Xi)(x)$ is a constant. 
In this case, we will simply write $\cW(\Phi,\Xi)$ without specifying the argument.
Besides, if $\Phi,\Xi\in\cD(L^{\max})$, then the Wronskian, a priori defined in $]a,b[$, can be extended to the endpoints:
\begin{equation}
	\cW(\Phi,\Xi)(a):=	\lim_{x\searrow a}\cW(\Phi,\Xi)(x),\quad \cW(\Phi,\Xi)(b):=	\lim_{x\nearrow b}\cW(\Phi,\Xi)(x).
\end{equation}
We also can equip  $\cD(L^{\max}) / \cD(L^{\min})$ with a non-degenerate bilinear antisymmetric form:
\begin{equation}\label{lagrange}
	\cW(\Phi,\Xi)(a)-\cW(\Phi,\Xi)(b)=-(\bar{L\Phi}|\Xi)+(\bar\Phi|L\Psi).
\end{equation}
(The identity in \eqref{lagrange} is sometimes called the {\em Lagrange identity} or then {\em integrated Green identity}.)

If \eqref{endpo}$=0$, then there exists a unique closed realization of $L$, which coincides with $L^{\max}=L^{\min}$. 
Otherwise $\sigma(L^{\max})=\sigma(L^{\min})=\cc$, and therefore the operators $L^{\max}$ and $L^{\min}$ are not well-posed. 
In order to define operators that may have a non-empty resolvent set, one needs to  select a space $\cD(L_\bullet)$ such that
\begin{align} 
	\cD(L^{\min})\subset\cD(L_\bullet)
	&\subset\cD(L^{\max}),\\
	\dim\cD(L^{\max})/\cD(L_\bullet)
	&=\dim\cD(L_\bullet)/\cD(L^{\min})
	= \frac{1}{2}(\nu_a+\nu_b).\label{endpo1}
\end{align}
Then we set $L_\bullet:=L^{\max}\Big|_{\cD(L_\bullet)}$.

To do this it is convenient to introduce the {\em boundary space}
\begin{equation}
	\cB:=\big(\cD(L^{\max})/\cD(L^{\min})\big)',\label{boundary}
\end{equation}
where the prime denotes the dual. Clearly $\dim\cB=2\nu_a+2\nu_b$.

If $\dim\cB=2$, in order to define $L_\bullet$ satisfying \eqref{endpo1} we need to fix a single nonzero functional $\phi_\bullet\in\cB$ to define the domain of an operator
\begin{equation}
	\cD(L_\bullet)=\{\Xi\in\cD(L^{\max})\ |\ \phi_\bullet(\Xi)=0\} .\label{wronski0}
      \end{equation}
      
This corresponds to two possibilities. If $\nu_a(L)=2$ and $\nu_b(L)=0$, a convenient way to define a functional $\phi_\bullet$ is to choose $\Phi_a\neq0$,  which near $a$ belongs to $\cD(L^{\max})$ but does not belong to $\cD(L^{\min})$, and then to set
\begin{equation}
	\phi(\Xi):=\cW(\Phi_a,\Xi)(a)=0 .\label{wronski}
\end{equation}
A good choice for  $\Phi_a$ is an element of $\cU_a(z)$ for some $z\in\cc$. 
(Usually, $z=0$ is  most convenient). 
The condition \eqref{wronski} will be called {\em the boundary condition (b.c.) at $a$ set by $\Phi_a$.}

Note that what is important in \eqref{wronski} is a nonzero functional on $\cD(L^{\max})$ vanishing on $\cD(L^{\min})$, which depends only on the behavior of $\Xi$ near $a$.
Sometimes for this end it is convenient to use a well chosen $\Phi_a$, which does not belong to $\cD(L^{\max})$.  We will see an example of this in the definition of the Whittaker operator.

Similarly, we proceed if $\nu_a(L)=0$ and $\nu_b(L)=2$. 
We  select $\Phi_b\neq0$ which near $b$ belongs to $\cD(L^{\max})$ but not $\cD(L^{\min})$, and set
\begin{equation}
	\phi_\bullet(\Xi):= \cW(\Phi_b,\Xi)(b).
\end{equation}

In our paper, we will not consider operators with $\dim\cB=4$, that is, $\nu_a(L)=\nu_b(L)=2$.
Nevertheless, for completeness let us discuss briefly this case.
To define a realization satisfying \eqref{endpo1}, we need to fix two linearly independent functionals $\phi_{\bullet 1},\phi_{\bullet 2}\in\cB$ and set
\begin{equation}
	\cD(L_\bullet)=\big\{\Xi\in\cD(L^{\max})\ |\ \phi_{\bullet1}(\Xi)=\phi_{\bullet2}(\Xi)=0\big\} .
\end{equation}
In particular, often one considers the so-called {\em separated boundary conditions}, where $\phi_{\bullet1}$ is a b.c.\ at $a$ and  $\phi_{\bullet1}$ is a b.c.\ at $b$.

\subsection{Candidate for resolvent}

Suppose now that
we want to describe the spectrum and resolvent of $L_\bullet$. 
First consider the case $\nu_a(L)=\nu_b(L)=0$.
Suppose $z\in\cc$ satisfies $\dim\cU_a(z)\geq 1$ and $\dim\cU_b(z)\geq 1$. 
We select 
\[
	\Psi_a(z,\cdot)\in\cU_a(z) \backslash \{0\},\quad 
	\Psi_b(z,\cdot)\in\cU_b(z)\backslash\{0\},\quad 
	\text{and set}\quad 
	\cW(z):=\cW\big(\Psi_b(z,\cdot),\Psi_a(z,\cdot)\big).   
\]
Then, we define
\begin{align}\label{eq:k}
	R_\bullet(z;x,y):= &\frac{1}{\cW(z)}
	\begin{cases}
		\Psi_a(z,x)\;\! \Psi_b(z,y) & \hbox{ if }  a<x < y<b, \\
		\Psi_a(z,y)\;\! \Psi_b(z,x) & \hbox{ if }  a<y < x<b. 
	\end{cases}
\end{align}

If $\nu_a(L)=2$, $\nu_b(L)=0$, then we need to select $\Phi_a$, 
setting the b.c.\ at  $a$,  and we define
\begin{equation} \label{bc-a}
	\cD(L_\bullet)=\{\Xi\in\cD(L^{\max})\ |\ \cW(\Xi,\Phi_a)(a)=0\} . 
\end{equation}
Similarly, if $\nu_a(L)=0$, $\nu_b(L)=2$, then we need to select $\Phi_b$, setting the b.c.\ at  $b$, so that
\begin{equation} \label{bc-b}
	\cD(L_\bullet)=\{\Xi\in\cD(L^{\max})\ |\ \cW(\Xi,\Phi_b)(b)=0\} . 
\end{equation}
If $\nu_a(L)=2$, $\nu_b(L)=2$ and we use separated b.c, then
\begin{equation} \label{bc-ab}
	\cD(L_\bullet)=\{\Xi\in\cD(L^{\max})\ |\
        \cW(\Xi,\Phi_a)(a)= \cW(\Xi,\Phi_b)(b)=0\} .
\end{equation}

Note that $\Psi_a$ and $\Psi_b$ above are defined uniquely up to a multiplicative constant.
$R_\bullet(z;x,y)$ does not depend on this choice.
The operator $R_\bullet(z)$ defined by the kernel $R_\bullet(z;x,y)$  sends $C_\mathrm{c}(]a,b[)$ into functions  in $L^2(]a,b[)$ satisfying the b.c.. 
Besides, we have
\begin{equation}
	\big(-\partial_x^2+V(x)-z\big)R_\bullet(z;x,y)=\big(-\partial_y^2+V(y)-z\big)R_\bullet(z;x,y)=\delta(x-y).
\end{equation}

The following theorem  is proven in \cite[Propositions 7.8 and 7.9]{DeGe}:
\begin{theoreme}  
	The following conditions are equivalent:
	\begin{enumerate}
		\item ${\dim}\,\cU_a(z)\geq1$, $\dim\, \cU_b(z)\geq1$ and $R_\bullet(z)$ is bounded;
		\item $z\in\cc\backslash\sigma(L_\bullet)$.
	\end{enumerate}
	If the above conditions hold, then
	\begin{equation}
		(L_\bullet-z)^{-1}(x,y)=	R_\bullet(z;x,y).\end{equation}
\end{theoreme}

Hence, the strategy for studying 1-dimensional Schr\"odinger operators on $L^2(]a,b[)$ given by the expression \eqref{oper1} is the following:
\begin{enumerate}
	\item For each $z\in\cc$, determine $\cN(L-z)$.
	\item Determine $\cU_a(z)$ and $\cU_b(z)$.
	\item If  $\nu_a(L)=2$ or $\nu_b(L)=2$, fix  b.c defining $L_\bullet$.
	\item For these b.c.\ (if needed), 
	and if ${\dim}\,\cU_a(z)\geq1$ and $\dim\, \cU_b(z)\geq 1$,
	write down $R_\bullet(z)$, using equation \eqref{eq:k},
	which is a candidate for the resolvent $(L_\bullet-z)^{-1}$.
	\item Check whether $R_\bullet(z)$ is bounded. 
	If so, $z\in\cc\backslash\sigma(L_\bullet)$ and $R_\bullet=(L_\bullet-z)^{-1}$.
\end{enumerate}

\subsection{Boundedness of resolvent}
\label{Boundedness of resolvent}

Let us quote two useful lemmas for proving the boundedness of an operator $K$ on $L^2(]a,b[)$ given by the integral kernel $K(x,y)$.
\begin{lemma} 
	The Hilbert--Schmidt norm of $K$ is given by
	\begin{equation}
		\|K\|_2:=\sqrt{\Tr K^*K}=\Big(\int |K(x,y)|^2 \,\d x\d y\Big)^{\frac12},
	\end{equation}
	and we have $\|K\|\leq \|K\|_2$.
	\label{hilberschmidt}
\end{lemma}

\begin{lemma}[Special case of Schur's Criterion] \label{lem:schur}
	Suppose that
	\begin{equation}
		\sup\limits_{x\in\,\left]a,b\right[}\int_a^b|K(x,y)|\,\d y=c_1,\quad\sup_{y\in\,\left]a,b\right[}\int_a^b|K(x,y)|\,\d x=c_2.
	\end{equation}
	Then $\|K\|\leq\sqrt{c_1c_2}$.
\end{lemma}

Note that, if applicable, Lemma \ref{hilberschmidt} is superior to Lemma \ref{lem:schur}, because it estimates a stronger norm.
For translation-invariant operators, we can also use a special case of the Young inequality, which actually follows from Schur's criterion:
\begin{lemma} 
	If $K(x,y)=f(x-y)$, then
	\begin{equation} 
		\|K\|\leq\int |f(x)|\,\d x.
	\end{equation}
\end{lemma}

Let us make some remarks on how to check the boundedness of an operator of the form \eqref{eq:k}. 
More precisely, suppose that $R$ is an operator with the integral kernel
\begin{align}\label{eq:ker}
	R(x,y):= & 
	\begin{cases}
		\Psi_a(x)\;\! \Psi_b(y) & \hbox{ if }  a<x < y<b, \\
		\Psi_a(y)\;\! \Psi_b(x) & \hbox{ if }  a<y < x<b; 
	\end{cases}
\end{align}
where $\Psi_a$, resp. $\Psi_b$ is square integrable near $a$, resp. $b$.
Choose $c$ such that $a<c<b$. Then we can split the operator $R$ into the sum of four operators
\begin{equation}\label{split1}
	R=R_{--}+R_{-+}+R_{+-}+R_{++}
\end{equation}
with kernels
\begin{subequations}\label{split2}
\begin{align}
	R_{--}(x,y)&:=\theta(c-x)\theta(c-y)R(x,y),\\
	R_{+-}(x,y)&:=\theta(x-c)\theta(c-y)R(x,y),\\
	R_{-+}(x,y)&:=\theta(c-x)\theta(y-c)R(x,y),\\
	R_{++}(x,y)&:=\theta(x-c)\theta(y-c)R(x,y).
\end{align}
\end{subequations}

Now $R_{+-}$ and $R_{-+}$ are bounded because both $\|R_{+-}\|_2^2$ and  $\|R_{-+}\|_2^2$ can be estimated by
\begin{equation} 
	\int_a^c|\Psi_a(x)|^2\,\d x\;\int_c^b|\Psi_b(y)|^2\,\d y.
\end{equation}
Thus, we obtain the following criterion:

\begin{lemma}\label{lemma-bound}
	$R$ is bounded iff $R_{--}$ and $R_{++}$ are bounded. 
\end{lemma}

\subsection{Krein formula}

We will always use equation \eqref{eq:k} to construct the resolvent of a Schr\"odinger operator. 
In the literature, in similar situations, one often uses the so-called Krein formula and the formalism of {\em boundary triplets} \cite{Derkach,Boitsev}. 
In this subsection, we compare the two approaches. We will not use boundary triplets in this paper;
however, we  discuss them briefly for the convenience of the reader.

Suppose, for definiteness, that $\nu_a(L) = 2$ and $\nu_b(L) = 0$. 
Suppose we fix a basis $\{\phi^0,\phi^1\}$ of the boundary space $\cB$ defined in \eqref{boundary}. 
(This is equivalent to fixing two distinct b.c.\ at $a$.) 
Then, for any $\kappa \in \cc \cup \{\infty\}$, we can define a closed realization $L_\kappa$ of $L$ by setting
\begin{align}
	\cD(L_\kappa)&:=\{\Xi\in\cD(L^{\max})\ |\
	\phi^0(\Xi)+\kappa\phi^1(\Xi)=0\},\quad\kappa\in\cc;\\
	\cD(L_\infty)&:=\{\Xi\in\cD(L^{\max})\ |\
	\phi^1(\Xi)=0\},\quad\kappa=\infty.
\end{align}
Let $z\in\cc$ and $0\neq\Psi_b(z,\cdot)\in\cU_b(z)$.
Note that under our assumptions $\Psi_b(z,\cdot)\in L^2(]a,b[)$.

Let $0\neq\Psi_a^i(z,\cdot)\in\cU_a(z)=\cN(L-z)$ with $\phi_a^i\big(\Psi_a^i(z,\cdot)\big)=0$, $i=0,1$.
Set
\begin{equation}
	\cW^i(z):=\cW\big(\Psi_b(z,\cdot),\Psi_a^i(z,\cdot)\big),\quad i=0,1.
\end{equation}
It follows from \eqref{eq:k} that
\begin{align}\label{eq:k.}
	R_\kappa(z;x,y):= &\frac{1}{\big(\cW^0(z)+\kappa\cW^1(z)\big)}
	\begin{cases}
		\big( \Psi_a^0(z,x)+\kappa\Psi_a^1(z,x)\big)\;\! \Psi_b(z,y) & \hbox{ if }  a<x < y<b, \\
		\big( \Psi_a^0(z,y)+\kappa\Psi_a^1(z,y)\big)\;\! \Psi_b(z,x) & \hbox{ if }  a<y < x<b. 
	\end{cases}
\end{align}
is a candidate for the kernel of $R_\kappa(z):=(L_\kappa-z)^{-1}$.
It is easy to check that
\begin{equation}
	\cW\Big(\Psi_a^1(z,\cdot)-\frac{\cW^1(z)}{\cW^0(z)}\Psi_a^0(z,x),\Psi_b(z,\cdot)\Big) =0.
\end{equation} 
Therefore, changing if needed the normalization of $\Psi_b(z,\cdot)$, we can assume that
\begin{equation}
	\Psi_a^1(z,\cdot)-\frac{\cW^1(z)}{\cW^0(z)}\Psi_a^0(z,\cdot)=\Psi_b(z,\cdot).
\end{equation}
Using this we can rewrite \eqref{eq:k.} as
\begin{equation}\label{krein}
	R_\kappa(z;x,y)=R_0(z;x,y)+\frac{\Psi_b(z,x) \Psi_b(z,y)}{\kappa^{-1}\cW^0(z)+\cW^1(z)}.
\end{equation}
This is often called {\em Krein formula}, which is the basis of the boundary triplet approach. It expresses the resolvent with mixed b.c.\ by the resolvent with unperturbed b.c.\ plus a rank one perturbation. Note that if we check the boundedness of $R_0(z)$, then $R_\kappa(z)$ is well-defined and  bounded unless $\kappa^{-1}\cW^0(z)+\cW^1(z)=0$.

Thus, in the above approach, one introduces three objects: the space $\cB$ and a pair of its distinguished linearly independent elements $\phi^0$ and $\phi^1$. 
Jointly, they are often called a {\em boundary triplet} \cite{Derkach,Boitsev}.
If $V$ is integrable near $a \in \rr$, one usually chooses the Dirichlet and Neumann b.c., that is, $\phi^0(\Xi) = \Xi(a)$ and $\phi^1(\Xi) = \Xi'(a)$. 
More generally, in most (but not all) Hamiltonians considered in this paper, we have  a similar distinguished pair. 
For example, for the Bessel operator with $|\Re(m)| < 1$, $m\neq0$, these are the b.c.\ set by $r^{\frac{1}{2}+m}$ and $r^{\frac{1}{2}-m}$. 
For $m = 0$, one can choose $r^{\frac{1}{2}}$ and $r^{\frac{1}{2}} \ln r$; however, one can argue that only the first is truly distinguished.

\section{Schr\"odinger operators related to the Bessel equation}
\label{sec:Bessel}

\subsection{Bessel operator}

Various elements of the material of this subsection can be found in the literature, e.g., in \cite{GTV,Kovarik,Pankrashkin}. 
We treat \cite{Derezinski2,Derezinski1} as the main references.

Let $m\in\mathbb{C}$. The {\em Bessel operator} is formally given by
\begin{align}\label{fact++}
	H_m:=-\partial_r^2 + \Big{(} m^2 - \frac{1}{4}\Big{)} \frac{1}{r^2}. 
\end{align}
We would like to interpret it as a closed operator on $L^2(\rr_+)$.

Note that when we deal with \eqref{fact++}, it is more convenient to use 1d Bessel functions rather than the (usual) 2d Bessel functions.
Depending on the circumstances, one may prefer to use hyperbolic or elliptic Bessel functions — therefore, for some quantities, we provide expressions in terms of both. 
For more detail about 1d Bessel functions, see Appendix \ref{Bessel equation}, especially Subsections \ref{Hyperbolic 1d Bessel equation} and \ref{Trigonometric 1d Bessel equation}.

In the following table, for each parameter and for each eigenvalue, we provide a few functions that span the space of eigenfunctions of \eqref{fact++} (usually, but not always, a basis of this space).
\[
\renewcommand{\arraystretch}{1.5}
\begin{array}{l|l|l}
	\hline
	\text{eigenvalue} & \text{parameters} & \text{eigenfunctions} \\
	\hline
	-k^2 \text{ with } k \neq 0& & \cI_{\pm m}(kx), \quad \cK_m(kx) \\
	\phantom{-}0 &m\neq0& x^{\frac{1}{2} \pm m}\\
	\phantom{-}0 &m=0& x^{\frac12},\quad
	x^{\frac12}\ln x\\
	\hline
\end{array}
\]
After checking the square integrability of these functions near the endpoints, we see that the  endpoints have the following types:
\[
\renewcommand{\arraystretch}{1.5}
\begin{array}{r|l|r}
	\hline
	\text{endpoint} &\text{parameters} & 
	\text{index} \\
	\hline
	0 & |\Re(m)|<1 &  2 \\
	0 & |\Re(m)| \geq 1  & 0 \\
	+\infty & &  0\\
	\hline
\end{array}
\]

\vspace{1em}

The following theorem describes the basic holomorphic family of Bessel operators:
\begin{theoreme}  
	For $\Re(m)\geq1$, there exists  a unique closed operator $H_m$ in the sense of $L^2(\rr_+)$, which on $C_\mathrm{c}^\infty(]0,\infty[)$ is given by
	\eqref{fact++}. 
	The family $m\mapsto H_m$ is holomorphic and possesses a unique holomorphic extension to $\Re(m)>-1$.

	The  spectrum and the point spectrum of $H_m$ are
	\begin{equation}
          \sigma(H_m)=[0,\infty[ \,,\quad\sigma_\mathrm{p}(H_m)=\emptyset\; ,
    \end{equation}
	and its resolvent is
	\begin{align}\label{reso1}
          (H_m+k^2)^{-1}
		(x,y)&=
		\frac{1}{k}\left\{\begin{matrix}\! \cI_m(kx)\;\! \cK_m(ky) & \hbox{ if } 0 < x < y, \\
			\cI_m(ky)\;\! \cK_m(kx) & \hbox{ if } 0 < y < x,
		\end{matrix}\right.&\Re(k)>0;\\
		(H_m-k^2)^{-1}(x,y)&=\pm \frac{\i}{k}\left\{\begin{matrix}\! \cJ_m(kx)\;\!
			\cH_m^\pm(ky) & \hbox{ if } \ 0 < x < y, \\
			\cJ_m(ky)\;\! \cH_m^\pm (kx) & \hbox{ if } \ 0 < y < x,
		\end{matrix}\right.&\pm\Im(k)>0.
	\end{align}
        \label{besseloperator}
      \end{theoreme}

\begin{proof}
For $\Re(m) > -1$, we define $H_m$ to be the closed realization of \eqref{fact++} with the b.c.\ at 0 given by $x^{\frac{1}{2} + m}$. From the table above, we see that for $\Re(m) \geq 1$, it is the unique realization of \eqref{fact++}.

We check that, for such $\Re(m)>-1$,
\begin{equation}
	\lim_{x\to0}  \cW(\cI_m(kx),x^{\frac12+m})=0,
\end{equation}
and $\cI_m(kx)$ is square integrable near zero.
Moreover, $\cK_m(kx)$ is square integrable near $+\infty$ if and only if $\Re(k)>0$. We also find
\begin{equation}
	\cW\big(\cI_m(kx),\cK_m(kx))=k.
\end{equation}
Next we apply \eqref{eq:k}, which yields the kernel on the rhs of \eqref{reso1} as a candidate for the resolvent of $H_m$.
We check that it is bounded. Hence it equals $(H_m+k^2)^{-1}$.
We also verify that it depends holomorphically on $m$ on a larger domain $\{\Re(m)>-1\}$.
Therefore, $\{\Re(m)>-1\}\ni m\mapsto H_m$ is a holomorphic family.

See \cite{Derezinski2} for the details.
\end{proof}

 Here is a description of the basic family of Bessel operators with real $m$. Again, its proof can be found, e.g., in \cite{Derezinski2}.
\begin{theoreme}
	For $m \in \left]-1,\infty\right[$, the operator $H_m$ defined in Theorem \ref{besseloperator} is self-adjoint. For $m \in [1,\infty[$, $H_m$ is essentially self-adjoint on $C_\mathrm{c}^\infty(]0,\infty[)$. For $m \in \left[0,\infty\right[$, it is the Friedrichs extension of \eqref{fact++} restricted to $C_\mathrm{c}^\infty(\left]0,\infty\right[)$. For $m \in \left]-1,0\right]$, it is the Krein extension of \eqref{fact++} restricted to $C_\mathrm{c}^\infty(\left]0,\infty\right[)$.
	\label{besseloperator2}
\end{theoreme}

Let us go back to complex $m$ satisfying $\Re(m)>-1$.
The Bessel operator is exactly solvable in a very strong sense. Besides the resolvent, one can also compute the integral kernel of the holomorphic semigroup generated by $H_m$, see, e.g., \cite{Kovarik}. It can also be written in two equivalent ways.
\begin{align}
	\e^{-\frac{t}{2}H_m}(x,y)
	&=\sqrt{\frac{2}{\pi t}}\,\cI_m\Big(\frac{xy}{t}\Big)\,\e^{-\frac{x^2+y^2}{2t}},&\Re(t) \geq0;\\  
	\e^{\pm\frac{\i t}{2}H_m}(x,y)
	&=\e^{\pm\i\frac{\pi}{2}(m+1)}\sqrt{\frac{2}{\pi t}}\,\cJ_m\Big(\frac{xy}{t}\Big)\,\e^{\frac{\mp\i x^2\mp\i y^2}{2t}},&\pm\Im(t)\geq0.\label{newa}
\end{align}

As noted in \cite{Derezinski2,Pankrashkin}, we have the identity
\begin{equation}
	H_m = \Xi_m^{-1}(A) K^{-1}\Xi_m(A),\label{diago}
\end{equation}
where
\begin{align}
	\Xi_m(t)&:=\e^{\i\ln(2)t}\frac{\Gamma(\frac{m+1+\i t}{2})}{\Gamma(\frac{m+1-\i t}{2})},\quad
	A:=\frac{1}{2\i}(x\partial_x+\partial_xx),\quad\text{and}\quad
	K:=x^2.
\end{align}
The operator $\Xi_m(A)$ is called the {\em Hankel transformation}. Its integral kernel  can be also computed:
\begin{equation}
	\Xi_m(A)(x,y)=\sqrt{\frac{2}{\pi}}\cJ_m(xy).
\end{equation}
Note that \eqref{diago} describes diagonalization of the Bessel operator --- transforming it into the multiplication operator $K^{-1}$.

\begin{remark}
For $-1<\Re(m)<1$, we can also consider mixed b.c..  For more details, see e.g. \cite{Derezinski1}.
\end{remark}

\subsection{Exponential potentials}\label{sec:exp}

For $k\in\mathbb{C}$, the {\em Schr\"odinger operator with the  exponential potential} is
formally given by
\begin{equation} \label{expo}
	M_k:=-\partial_x^2+k^2\e^{2x}. 
\end{equation}
We will interpret it as a closed operator on $L^2(\rr)$.
Without restricting the generality we can assume that $\Re(k)\geq0$. 

Note that $M_k$ is very common in the literature. 
For example, it has important applications in Liouville CFT, see, e.g., \cite[Chapter 4]{Seiberg}, and on hyperbolic manifolds.

For $r=\e^x$, we have the following formal identity: 
\begin{equation} \label{eq:exp-formal-id}
	r^2\Big(-\partial_r^2+\big(m^2-\tfrac14\big)\frac1{r^2}+k^2\Big)
	=\e^{\frac{x}{2}}\Big(-\partial_x^2+k^2\e^{2x}+m^2\Big)\e^{-\frac{x}{2}}. 
\end{equation} 
Using this, we can express  eigenfunctions of $M_k$ in terms of Bessel functions. In this case, 2-dimensional Bessel functions are more convenient than 1-dimensional ones, see Appendix, Section \ref{Hyperbolic 2d Bessel equation}.

Eigenvalues and corresponding eigenfunctions of \eqref{expo} are given by:
\[
\renewcommand{\arraystretch}{1.5}
\begin{array}{l|l|l}
	\hline
	\text{eigenvalue} & \text{parameters}&\text{eigenfunctions} \\
	\hline
	-m^2 &\Re(k)\geq0,\quad  k \neq 0 & I_{\pm m}(k\e^x), \quad K_m(k\e^x) \\
	-m^2 \quad m\neq0& k = 0 & \e^{\pm mx}\\                                         
	\phantom{-}0 \quad& k = 0 & 1,\quad x \\
	\hline
\end{array}
\]
After checking the square integrability of these functions near the endpoints, we see that the  endpoints have the following indices:
\[
\renewcommand{\arraystretch}{1.5}
\begin{array}{r|l|r}
	\hline
	\text{endpoint} & 
	\text{parameters} & \text{index} \\
	\hline
	-\infty &\Re(k)\geq0 &  0 \\
	+\infty & \Re(k)>0 \text{ or } k=0  & 0 \\
	+\infty & \Re(k)=0, \quad k\neq 0  & 2 \\
	\hline
\end{array}
\]

\vspace{1em}

The following theorem describes the basic holomorphic family of Schr\"odinger operators with the exponential potential:

\begin{theoreme} \label{thm:exponential} 
	For $\Re(k)>0$ or $k=0$, the expression \eqref{expo} defines a unique closed operator on $L^2(\rr)$, which will be denoted $M_k$.
The spectrum and point spectrum are  $\sigma(M_k)=[0,\infty[$, $\sigma_\mathrm{p}(M_k)=\emptyset$.
	
	Moreover, $\{\Re(k)>0\}\ni k\mapsto
	M_k$ is a holomorphic family of closed operators, and
	for $\Re(m)>0$,  the resolvent is given by
	\begin{align}\label{eq:resolventofMk}
		(M_k+m^2)^{-1}(x,y)&= 
		\begin{cases}
			I_m(k\e^x)\;\! K_m(ke^y) & \hbox{ if }  x < y, \\
			I_m(k\e^y)\;\! K_m(k\e^x) & \hbox{ if }  y < x.
		\end{cases}
	\end{align}
\end{theoreme}

\begin{proof}
The case $k=0$ is well known, let us restrict ourselves to $\Re(k)>0$.
We check that for $\Re(m)>0$ $I_m(k\e^x)$ is up to a multiplicative constant the only eigensolution square integrable close to $-\infty$. Similarly for $K_m(k\e^x)$  close to $+\infty$. 
We check that
\begin{equation}
	\cW\big(I_m(k\e^x),K_m(k\e^x))=1.
\end{equation}
Next, we apply \eqref{eq:k}, which yields the kernel on the rhs of \eqref{eq:resolventofMk} as a candidate for the resolvent of $M_k$. 
Let us denote it by $R_k(-m^2;x,y)$.
In Lemma \ref{lem:holomorphic-Rkm2}, following the strategy given in Section \ref{Boundedness of resolvent}, we check the boundedness of $R_{k}(-m^{2})$ and its analyticity wrt $k$. This ends the proof of the theorem.
\end{proof}

\begin{lemma}[Boundedness of Kernel]\label{lem:holomorphic-Rkm2} 
	Let $\Re(m)>0$ and $\Re(k)>0$.
	Let $R_{k}(-m^{2})$ be the operator with kernel \eqref{eq:resolventofMk}.
	Then $R_{k}(-m^{2})$ is a bounded operator and the map $k\mapsto R_{k}(-m^{2})$ is a holomorphic family of bounded operators, which does not have a holomorphic extension to a larger subset of the complex plane.
	\label{popo3}
\end{lemma}

\begin{proof}
Since (a) both modified Bessel functions $ I_{m}$ and $  K_{m}$ are analytic for fixed $m$ with $\Re(m)>0$ and (b) the function $k\e^{x}$ is analytic in $k$, the kernel $R_{k}(-m^{2};x,y)$ is an analytic function of parameter $k$. 
It is easy to see that for any $f,g\in C_\mathrm{c}^{\infty}(\rr)$, the quantity 
\[
	( f | R_{k}(-m^{2})\,g )
	=\int\overline{f(x)} \, R_{k}(-m^{2};x,y) \, g(y)\,\mathrm{d}x\mathrm{d}y
\]
is analytic in $k$. 
Since $C_\mathrm{c}^{\infty}(\rr)$ is a dense subset of $L^{2}(\rr)$, it remains to prove that $R_{k}(-m^{2};x,y)$ is locally bounded in $k$.

To proceed further, we use the method of Section \ref{Boundedness of resolvent}.  
We split the resolvent as in \eqref{split1} and \eqref{split2} with $c=0$. 
By Lemma \ref{lemma-bound}, we need to prove the boundedness of $R_{--}$ and $R_{++}$.

We have, for $x\to0$,
\begin{align*}
	I_{m}(x) & \sim\frac{1}{\Gamma(m+1)}\left(\frac{x}{2}\right)^{m},\quad m\neq-1,-2,\dots,\\
	K_{m}(x) & \sim
	\begin{cases}
		\Re(\Gamma(m)\left(\frac{2}{x}\right)^{m}) & \text{if }\Re(m)=0,m\neq0,\\
		-\ln\left(\frac{x}{2}\right)-\gamma & \text{if }m=0,\\
		\frac{\Gamma(m)}{2}\left(\frac{2}{x}\right)^{m} & \text{if }\Re(m)>0,\\
		\frac{\Gamma(-m)}{2}\left(\frac{2}{x}\right)^{m} & \text{if }\Re(m)<0,
	\end{cases}
\end{align*}
and, for $x\to\infty,$
\[
	I_{m}(x)\sim\frac{1}{\sqrt{2\pi x}}\e^{x}\quad\text{and}\quad  K_{m}(x)\sim\sqrt{\frac{\pi}{2x}}\e^{-x}.
\]
We check
\begin{align*}
	&\left|R_{k}^{--}(-m^{2};x,y)\right|\\
	&\quad\leq C_{m,k}\left(\e^{-\Re(m)x}\e^{\Re(m)y}\mathds{1}_{-\infty<y<x<0}(x,y)+\e^{\Re(m)x}\e^{-\Re(m)y}\mathds{1}_{-\infty<x<y<0}(x,y)\right).
\end{align*}
Since
\begin{align*}
	&\sup_{x\in\,]-\infty,0]}\int\e^{-\Re(m)x}\e^{\Re(m)y}\mathds{1}_{-\infty<y<x<0}(x,y)\,\mathrm{d}y \\
	&\quad=\sup_{x\in\,]-\infty,0]}\e^{-\Re(m)x}\int_{-\infty}^{x}\e^{\Re(m)y}\,\mathrm{d}y=\sup_{x\in\,]-\infty,0]}\e^{-\Re(m)x}\frac{1}{\Re(m)}\e^{\Re(m)x}=\frac{1}{\Re(m)},\\
	&\sup_{y\in\,]-\infty,0]}\int\e^{-\Re(m)x}\e^{\Re(m)y}\mathds{1}_{-\infty<y<x<0}(x,y)\,\mathrm{d}x \\
	&\quad=\sup_{y\in\,]-\infty,0]}\e^{\Re(m)y}\int_{y}^{0}\e^{-\Re(m)x}\,\mathrm{d}x=\sup_{y\in\,]-\infty,0]}\e^{\Re(m)y}\frac{1}{\Re(m)}(-\int_{0}^{-y}\e^{\Re(m)x}\,\mathrm{d}x)=\frac{1}{\Re(m)},
\end{align*}
by Schur's criterion we obtain 
\[
	\left\Vert R_{k}^{--}(-m^{2})\right\Vert \leq2C_{m,k}.
\]

For the $R^{++}_{k}(-m^{2})$, we have the following: Since $\Re(k) > 0$, we obtain
\[
	\left|R^{++}_{k}(-m^{2};x,y)\right|\leq C_{m,k}\e^{-\frac{x+y}{2}}\e^{-\Re(k)|\e^{x}-\e^{y}|}
	\leq C_{m,k}\e^{-\frac{x+y}{2}}.
\]
Therefore, the Hilbert--Schmidt norm of $R^{++}_{k}(-m^{2})$ is finite.

Now, we prove that $R_{k}(-m^{2})$ cannot be extended to a holomorphic family of bounded operators beyond the axis $\ensuremath{\Re(k)=0}$.
Let us fix $g\in C_\mathrm{c}^{\infty}(\rr)$.
For $k\mapsto R_{k}(-m^{2})g$, with values in $L_{\mathrm{loc}}^{2}(\rr)$ is entire analytic.
If $R_{k}(-m^{2})$ could be extended to a holomorphic family of bounded operators, when applied to the function $g$ this extension should coincide with $R_{k}(-m^{2})g$. 
Then, for $x$ below the support of $g$, we have 
\[
	R_{k}(-m^{2})g(x)=\int  I_{m}(k\e^{x}) K_{m}(k\e^{y})g(y)\,\mathrm{d}y
	= I_{m}(k\e^{x})\int K_{m}(k\e^{y})g(y)\,\mathrm{d}y=C_{m,k} I_{m}(k\e^{x})
\]
and
\[
	R_{k}(-m^{2})g(x)=\int  K_{m}(k\e^{x}) I_{m}(k\e^{y})g(y)\,\mathrm{d}y
	= K_{m}(k\e^{x})\int  I_{m}(k\e^{y})g(y)\,\mathrm{d}y=C_{m,k} I_{m}(k\e^{x})
\]
for some constant $C_{m,k}$.

If $\Re(k)< 0$, then $ I_{m}(k\e^{x})\not\in L^2(\rr)$  because $ | I_{m}(x)| $ diverges as $|\Re(x)|\to \infty$.			
Hence, the map cannot be extend to $\Re(k)<0$.
\end{proof}

\begin{remark}
	One of applications of perturbed Bessel operators to quantum physics
	is the concept of {\em Regge poles} \cite{Regge}. 
	They are defined as  poles of the holomorphic function $m\mapsto (H_m + V + k^2)^{-1}$, where $V$ is typically a short-range potential.
	Substituting $r=\e^x$ similarly to \eqref{eq:exp-formal-id}, we obtain
	\[
	r^2 (H_m + V(r) + k^2 ) = \e^{\frac{x}{2}} \Big( -\partial_x^2 + k\e^x + \e^{2x} V(e^x) + m^2 \Big) \e^{-\frac{x}{2}}.
	\]
	This substitution is quite useful in many cases and sometimes called the {\em Langer substitution} \cite[Eq. (22)]{Langer}.
	Set $W(x):=\e^{2x}V(\e^x)$.
	Then, the transmutation property \eqref{regge} can be generalized to include a potential:
	\begin{equation}
		(M_k+W+m^2)^{-1}(x,y)
		=\e^{-\frac{x}{2}}(H_m+V+k^2)^{-1}(\e^x,\e^y)\,\e^{-\frac{y}{2}}.
	\end{equation}
	Thus Regge poles coincide with the poles of $m\mapsto(M_k + W + m^2)^{-1}$.
	Therefore, we have another equivalent definition 
	of  Regge poles, which is used e.g. in \cite{BBD}.
\end{remark}

\begin{figure}[ht]
	\centering
	\begin{subfigure}[b]{0.48\textwidth}
		\includegraphics[width=\textwidth]{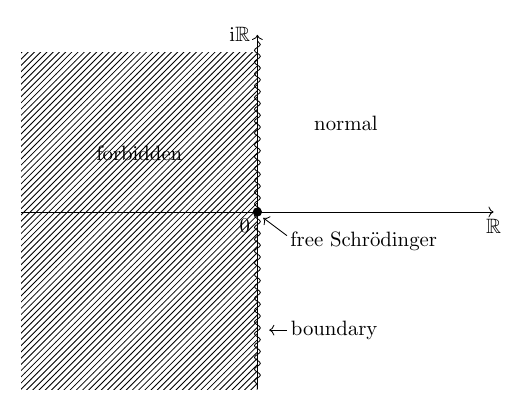}
		\caption{For $M_k:= -\partial_x^2 + k^2 e^{2x}$ with $k \in \mathbb{C}$}
		\label{fig:exp_plot_k}
	\end{subfigure}
	\begin{subfigure}[b]{0.48\textwidth}
		\centering
		\includegraphics[width=\textwidth]{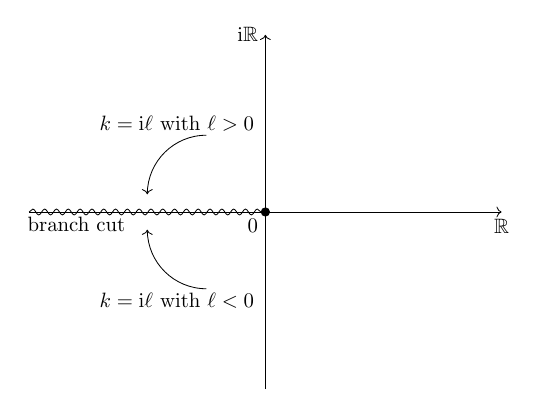}
		\caption{For $-\partial_x^2 + c \, e^{2x}$ with $c \in \mathbb{C}$ }
		\label{fig:exp_plot_c}
	\end{subfigure}
	\caption{Domains of the parameters $k$ and $c=k^2$}
\end{figure}

\subsection{Negative exponential potential}\label{sec:negative-exp}

The previous subsection covered  the case $\Re(k)>0$. 
In this subsection, we consider the case $\Re(k)=0$, $k\neq0$, that is, the {\em Schr\"odinger operator with a negative exponential potential}. 
Surprisingly, it appears in interesting physical applications, e.g., it is the main ingredient of the Feynman propagator on the Poincaré patch of the de Sitter and anti-de Sitter space. 
Clearly, it defines a  Hermitian operator, which possesses a 1-parameter family of closed realizations on $L^2(\rr)$.

In this section, for convenience, we introduce the parameter $\ell > 0$, so that $k =\ii\ell$, and we consider the formal expression 
\begin{equation} \label{eq:neg-exp}
	M_{\ii\ell} := -\partial_x^2 - \ell^2 \e^{2x}.
\end{equation} 
The corresponding maximal and minimal operators of $M_{\ii\ell}$ in $L^2(\rr)$ are denoted $M_{\ii\ell}^{\max}$ and $M_{\ii\ell}^{\min}$. The domain of $M_{\ii\ell}^{\max}$ is given by
\[
	\cD(M_{\ii\ell}^{\max}) = \{ f \in L^2(\rr) \, | \, (-\partial_x^2 - \ell^2 \e^{2x}) f \in L^2(\rr) \},
\]
and $M_{\ii\ell}^{\min}$ is the closure of the restriction of \eqref{eq:neg-exp} to $C^\infty_{\mathrm{c}}(\rr)$.

In order to set b.c., we will use the Hankel functions
\[
	H_{\frac{1}{2}}^\pm(r) = \mp\ii\Big(\frac{2}{\pi r}\Big)^{\frac{1}{2}}\e^{\pm\ii r}.
\]
(We could use $H_m^\pm(r)$ with other $m$, but the  parameter $\frac{1}{2}$ gives especially simple elementary functions.)

First, we describe the two distinguished realizations:
\begin{theoreme} 
	Let $\ell>0$. Then there exists two
	closed operators in the sense of $L^2(\rr)$ that on
	$C_\mathrm{c}^\infty(]0,\infty[)$ is given by \eqref{eq:neg-exp}
	and satisfy the following b.c.\ at $+\infty$:  
	\begin{align}
		\cD(M^0_{\ii\ell}) = \Big\{ \Xi \in \cD(M_{\ii\ell}^{\max})
		\mid \lim_{x \to \infty} \cW\big(H_{\frac12}^+(\ell \e^x) \, , \,
		\Xi(x)\big) = 0 \Big\},\\
		\cD(M^\infty_{\ii\ell}) = \Big\{ \Xi \in \cD(M_{\ii\ell}^{\max}) \mid \lim_{x \to \infty} \cW\big(H_{\frac12}^-(\ell \e^x) \, , \, \Xi(x)\big) = 0 \Big\}. 
	\end{align} 
	Both $M_{\ii\ell}^0$ and $M_{\ii\ell}^\infty$ do not have point spectrum, more precisely,
	\begin{equation} 
		\sigma(M_{\ii\ell}^0) = \sigma(M_{\ii\ell}^\infty) =
                [0,\infty[\,, \quad 
                		\sigma_\mathrm{p}(M_{\ii\ell}^0) = \sigma_\mathrm{p}(M_{\ii\ell}^\infty) = \emptyset
                \; , 
	\end{equation}
	and, for $\Re(m)>0$,
	\begin{align}\label{popo1}
	(M_{\ii\ell}^0+m^2)^{-1}(x,y)
		&=\frac{\pi\i}{2}
		\begin{cases}
			J_m(\ell \e^x)\;\! H_m^+(\ell \e^y)& \hbox{ if } x < y, \\[2ex]
			J_m(\ell \e^y)\;\!H_m^+(\ell \e^x)& \hbox{ if } y < x;
		\end{cases}\\\label{popo2}
		(M_{\ii\ell}^\infty+m^2)^{-1}(x,y)
		&=-\frac{\pi\i}{2}
		\begin{cases} 
			J_m(\ell \e^x)\;\!  H_m^-(\ell \e^y)& \hbox{ if } x < y, \\[2ex]
			J_m(\ell \e^y)\;\!  H_m^-(\ell \e^x)& \hbox{ if } y < x.
		\end{cases}
	\end{align}
\end{theoreme}

\begin{proof}
	First we check that 
	\begin{equation}
		\lim_{x\to\infty}\cW\big(H_{\frac12}^\pm(\ell \e^x), H_{m}^\pm(\ell\e^x)\big)
		=0.
	\end{equation}
	Therefore, the rhs of \eqref{popo1}, resp. \eqref{popo2} are good candidates for the inverses of $M_{\ii\ell}^\infty+m^2$, resp. $M_{\ii\ell}^0+m^2$. But  the rhs of \eqref{popo1}, resp. \eqref{popo2} coincide with $R_k(-m^2)$ for $k=\ii\ell$, resp. $k=-\ii\ell$, in the notation  of Lemma \ref{popo3} (see the proof of Thm \ref{thm:exponential-conv}). 
	And the proof of this lemma  applies. \end{proof}

In the following theorem, we prove that the realizations described above are limiting cases  of the basic holomorphic family from the previous subsection.

\begin{theoreme}\label{thm:exponential-conv}
	We have the following weak convergence: 
	\begin{enumerate}
		\item We have
		\begin{equation}
			\wlim_{\varepsilon\searrow0}(M_\varepsilon+m^2)^{-1}=(-\partial_x^2+m^2)^{-1}.
		\end{equation}
		\item For $\ell >0$ and $\Re(m)>0$, we have
		\begin{align}
			\wlim_{\varepsilon\searrow0}(M_{\varepsilon + \i \ell}+m^2)^{-1}=(M_{\ii\ell}^\infty+m^2)^{-1},\\
			\wlim_{\varepsilon\searrow0}(M_{\varepsilon - \i \ell}+m^2)^{-1}=(M_{\ii\ell}^0+m^2)^{-1}.
		\end{align}
	\end{enumerate}
\end{theoreme}

\begin{proof}
	Note that 
	\[
	J_m(\ell \e^x) = \e^{\pm \i \frac{\pi}{2}m} I_m(\mp \i \ell \e^x)
	\quad\text{and}\quad
	H^{\pm}_m(\ell \e^y) = \frac{2}{\pi} \e^{\mp \i \frac{\pi}{2} (m+1)} K_m(\mp \ell \e^y).
	\]
	Then, for all $x<y$, either
	\[
	I_m\left((\varepsilon + \i \ell)\e^x\right)K_m\left((\varepsilon + \i \ell)\e^y\right)
	\to
	\frac{\pi}{2} J_m(\ell \e^x)H^{-}_m(\ell \e^y)
	=
	I_m(\i \ell \e^x) K_m(\i \ell \e^y)
	\]
	or
	\[
	I_m\left((\varepsilon + \i \ell)\e^x\right)K_m\left((\varepsilon + \i \ell)\e^y\right)
	\to
	\frac{\pi}{2} J_m(-\ell \e^x)H^{+}_m(-\ell \e^y)
	=
	I_m(-\i \ell \e^x) K_m(-\i \ell \e^y)
	\]
	as $\varepsilon \to 0$.
	Similar convergence can be obtained for $y<x$.
	Then, as $\ell\neq 0$ with uniform boundness of $I_m$ and $K_m$ and Lebesgue dominate convergence theorem gives the norm convergence.
\end{proof}

From Theorem \ref{thm:exponential-conv}, we see that by setting
\begin{align}
	M_0:=-\partial_x^2,\quad
	M_{\i\ell}:=M_{\ii\ell}^\infty,\quad
	M_{-\i\ell}:=M_{\ii\ell}^0;\quad \ell>0,
\end{align}
we extend the basic holomorphic family $\{\Re(k)>0\}\ni k\mapsto M_k$ to a
continuous family $\{\Re(k)\geq0\}\ni k\mapsto M_k$.

Before we continue, let us note that  $J_m(\ell\e^x)$ belong to
$L^2(\rr)$ for $\Re(m)>0$. Indeed,
\begin{equation}
	\int_\rr |J_{m}(\ell \e^x)|^2 \,\d x =\int_\rr \frac{|J_{m}(y)|^2}{y} \,\d y
	= \frac{1}{2m} 
        ,\label{nistbook}
\end{equation}
where the last identity is e.g. in \cite[Eq. (10.22.57)]{NIST} with $a=1$, $\mu=\nu=m$, and $\lambda = 1$.

Now, we are ready describe the remaining mixed realizations (see \cite{Stempak} for a similar result).

\begin{theoreme}
	Let $\ell>0$ and $\gamma\in\cc\setminus\{0\}$. There exists a unique closed operator in the sense of $L^2(\rr)$ that on $C_\mathrm{c}^\infty(]0,\infty[)$ is given by \eqref{eq:neg-exp} and satisfies the following b.c.\ at $+\infty$:
	\begin{equation}
		\cD(M^\gamma_{\ii\ell}) = \Big\{ \Xi \in \cD(M_{\ii\ell}^{\max}) \mid \lim_{x \to \infty} \cW\big(H_{\frac12}^+(\ell \e^x) + \gamma H_{\frac12}^-(\ell \e^x) \, , \, \Xi(x)\big) = 0 \Big\}. 
	\end{equation} 
	$\cc\backslash\{0\}\ni\gamma\mapsto M^\gamma_{\ii\ell}$ is  a holomorphic family of closed operators.  
	If $\gamma=\e^{\i\alpha}$, the eigenvalues and eigenfunctions of $M_{\ii\ell}^\gamma$ are
	\begin{align} 
		-(\alpha+n)^2,
		&\quad J_{\alpha+n}(\ell\e^x),\quad n\in 2\zz,\quad \Re(\alpha+n)>0.
	\end{align}
	its spectrum, point spectrum  and  resolvent are
	\begin{align}
		\sigma(M_{\ii\ell}^\gamma)
		&=[0,\infty[\;\cup\;
           \sigma_\mathrm{p}(M_{\ii\ell}^\gamma),\\
          \sigma_\mathrm{p}(M_{\ii\ell}^\gamma)&= \{-(n+\alpha)^2\ | \; \Re(n+\alpha)>0,\quad\ n\in2\zz\},\\
	(M_{\ii\ell}^\gamma +m^2)^{-1}(x,y)
		&=\frac{\pi\i}{2(\e^{\i m\pi}-\gamma)}
		\begin{cases} 
			J_m(\ell \e^x)\;\! \big(\e^{\i m\pi} H_m^+(\ell \e^y)+\gamma H_m^-(\ell \e^y)\big)
			& \hbox{ if } x < y, \\[2ex]
			J_m(\ell \e^y)\;\! \big(\e^{\i m\pi} H_m^+(\ell \e^x)+\gamma H_m^-(\ell \e^x)\big)
			& \hbox{ if } y < x.
		\end{cases}\label{resol}
	\end{align}
\end{theoreme}

\vspace{1em}

\section{Schr\"odinger operators related  to the Whittaker equation}
\label{sec:whittaker}

\subsection{Whittaker operator}

This is another classic problem, described in many sources.
We treat \cite{DR, DFNR} as the main references for this subsection.

Let $m,\beta\in\mathbb{C}$. The \emph{Whittaker operator} is formally defined by 
\begin{align}\label{eq:fact+}
	H_{\beta,m}:=-\partial_r^2 + \big( m^2 - \tfrac{1}{4} \big) \frac{1}{r^2}-\frac{\beta}{r}. 
\end{align} It
is the radial part of the Schr\"odinger operator with the Coulomb potential, in dimension 3 used to describe the  Hydrogen atom. 
We will interpret it as a closed operator on $L^2(\rr_+)$.

We will use  1d Whittaker functions with conventions described in Appendix, Subsect. \ref{The 1d Whittaker equation}.

We first find eigenvalues and corresponding eigenfunctions of \eqref{eq:fact+}:
\[
\renewcommand{\arraystretch}{1.5}
\begin{array}{l|l|l}
	\hline
	\text{eigenvalue} &\text{parameters}& \text{eigenfunctions} \\
	\hline
	-k^2 \text{ with } k\neq0 & &\cI_{\frac{\beta}{2k},\pm m}(2kr), \quad \cK_{\frac{\beta}{2k},m}(2kr) \\
	\phantom{-}
	0  &\beta\neq0& r^{\frac{1}{4}}\cJ_{\pm
		2m}(2\sqrt{\beta r}), \quad
	r^{\frac{1}{4}}\cH_{2m}^\pm(2\sqrt{\beta
		r}) \\
	\phantom{-}  0&\beta=0,\quad m\neq0&r^{\frac12\pm m}\\
	\phantom{-}  0&\beta=0,\quad m=0&r^{\frac12},\quad r^{\frac12}\ln(r)\\
	\hline
\end{array}
\]
After checking the square integrability of these functions near the endpoints, we see that the endpoints have the following indices:
\[
\renewcommand{\arraystretch}{1.5}
\begin{array}{r|l|r}
	\hline
	\text{endpoint} & \text{parameters}
	& \text{index} \\
	\hline
	0 & |\Re(m)|<1 &  2 \\
	0 & |\Re(m)|\geq1  & 0 \\
	+\infty & &  0 \\
	\hline
\end{array}
\]

\vspace{1em}

Let us describe the basic holomorphic family of Whittaker operators:
\begin{theoreme} \label{thm:whittaker}
	For $\Re(m)\geq 1$, there exists a unique closed operator $H_{\beta,m}$ in the sense of $L^2(\rr_+)$, which on $C_\mathrm{c}^\infty(]0,\infty[)$ is given by \eqref{eq:fact+}. 
	It depends holomorphically on $\beta,m$. 
	It can be uniquely extended to a holomorphic family of closed operators on $L^2(\rr_+)$ 
	defined for $\Re(m)>-1$, $\beta\in\cc$,  $(\beta,m)\neq(0,-\frac12)$.
	Its spectrum and point spectrum are
	\begin{align}
		\sigma(H_{\beta,m}) &
		= [0,\infty[
		\; \cup \; 		\sigma_\mathrm{p}(H_{\beta,m}) \\
	\sigma_\mathrm{p}(H_{\beta,m}) &=		\Big\{
			-\frac{\beta^2}{4(n+m+\frac12)^2}\;\Big|\; n+m+\frac12\neq0,\;\Re\Big(\frac{\beta}{n+m+\frac12}\Big)>0,\; n\in\nn_0
		\Big\}.
	\end{align}
	
	Outside of the spectrum, the kernel of the resolvent of $H_{\beta,m}$ is
	\begin{equation}\label{The-resolvent}
		(H_{\beta,m}+k^2)^{-1}(x,y)
		:= \tfrac{1}{2k} \,
		\Gamma\big(\tfrac{1}{2}+m-\tfrac{\beta}{2k}\big)
		\begin{cases} \cI_{\frac{\beta}{2k},m}(2k x)\cK_{\frac{\beta}{2k},m}(2k y) & \mbox{ for }0<x<y,\\
			\cI_{\frac{\beta}{2k},m}(2k y)\cK_{\frac{\beta}{2k},m}(2k x) & \mbox{ for }0<y<x.
		\end{cases}
	\end{equation}
\end{theoreme}

\vspace{1em}

\begin{proof}[Proof of Theorem \ref{thm:whittaker}]
	We define $H_{\beta,m}$ for $\Re(m)>-1$, $(\beta,m)\neq(0,-\frac12)$ by the b.c.\ at $0$ given by $x^{\frac14}\cJ_{2m}(2\sqrt{\beta x})$. 
	We check that for $\Re(m)>-1$
	\begin{equation}
		\lim_{x\to0}  \cW\Big(\cI_{\frac{\beta}{2k},\pm m}(2kx),  x^{\frac14}\cJ_{2m}(2\sqrt{\beta x})\Big)=0.
	\end{equation}
	Moreover, $\cK_{\frac{\beta}{2k},m}(2kx)$ is square integrable near $+\infty$. We also find
	\begin{equation}
		\cW\Big(\cI_{\frac{\beta}{2k},m}(2kx),\cK_{\frac{\beta}{2k},m}(2k x)\Big)
		=\frac {2k}{\Gamma\big(\tfrac{1}{2}+m-\tfrac{\beta}{2k}\big)}.
	\end{equation}
	Next, we apply \eqref{eq:k}, which yields the kernel on the right-hand side of \eqref{reso1} as a candidate for the resolvent of \( H_{\beta,m} \).
	We check that it is bounded. Hence, it equals \( (H_{\beta,m} + k^2)^{-1} \). 
	We also verify that it depends analytically on \( \beta \) and \( m \). 
	Therefore,
	\begin{equation}
		\cc \times \{\Re(m) > -1\} \backslash (0, -\tfrac12)
		\ni (\beta, m) \mapsto H_{\beta,m}
	\end{equation}
	is an analytic family.
	For $\Re(m) \geq 1$, the b.c.\ is not needed, hence $H_{\beta,m}$ is then uniquely defined.
	
	See \cite{DR} for details.
\end{proof}

The operator $H_{\beta,m}$ is sometimes called the {\em Whittaker operator with pure b.c.}.
Note that for $m>-\frac12$, we can simplify the b.c.\ --- we can set it by $x^{\frac12+m}$. 
For $-1<\Re(m)\leq -\frac12$, interestingly, this does not work. 
Instead, we can use $x^{\frac12+m}(1-\frac{\beta}{1+2m})$ to set the b.c..
Details can be found in \cite{DR}. 

Note that $(0,-\frac{1}{2})$ is a singularity of the holomorphic function $(\beta,m) \mapsto H_{\beta,m}$. 
Additionally, we set $H_{0,-\frac{1}{2}} := H_{-\frac{1}{2}}$.
Then, the basic family of Whittaker operators extends the basic family of Bessel operators: 
\begin{equation} 
	H_{0,m}=H_m,\quad \Re(m)>-1.
\end{equation}

\begin{remark}
	For $-1<\Re(m)<1$, $\beta\in\cc$, we can also consider mixed  b.c.. We do not study them here, see e.g. \cite{DFNR}.
\end{remark}

\subsection{Morse potentials}\label{sec:Morse}

The {\em Schr\"odinger operator with the Morse potential} is formally given by 
\begin{equation}\label{eq:morse}
	M_{\beta,k}:=-\partial_x^2+k^2\e^{2x}-\beta\e^x. 
\end{equation}
We will interpret it as a closed operator on $L^2(\mathbb{R})$.
Without restricting the generality we can assume that $\Re(k)\geq0$.

For $r=\e^x$, we have the following formal identity: 
\begin{equation}
	r^2\Big(-\partial_r^2+\big(m^2-\tfrac14\big)\frac1{r^2}-\frac{\beta}{r}+k^2\Big)=\e^{\frac{x}{2}}\Big(-\partial_x^2+k^2\e^{2x}-\beta\e^x+m^2\Big)\e^{-\frac{x}{2}}.
\end{equation}
Therefore, eigenfunctions of $M_{\beta,k}$ can be expressed in terms of Whittaker functions. In this subsection instead of the standard (1-dimensional) Whittaker functions $\cI_{\beta,m},\cK_{\beta,m}$ it is more convenient to use  2d Whittaker functions $I_{\beta,m},K_{\beta,m}$, see Appendix \ref{sec:newWhittaker}.

We first find eigenvalues and corresponding eigenfunctions of \eqref{eq:isot}:
\[
\renewcommand{\arraystretch}{1.5}
\begin{array}{l|l|l}
	\hline
	\text{eigenvalue} &\text{parameters}& \text{eigenfunctions} \\
	\hline
	-m^2 &\Re(k)\geq0,\quad  k \neq 0 & I_{\frac{\beta}{2k},\pm m}(2k\e^{x}), \quad K_{\frac{\beta}{2k},m}(2k\e^{x}) \\
	-m^2 & k = 0, \quad\beta \neq 0 & I_{\pm 2m}(\sqrt{-4\beta}\, \e^{\frac{x}{2}}), \quad K_{2m}(\sqrt{-4\beta}\,  \e^{\frac{x}{2}} )\\
	-m^2 \text{ with }m\neq0&k = 0, \quad\beta=0 &  \e^{\pm  mx} \\
	\phantom{-}0&k=0,\quad\beta=0&1,\quad x\\
	\hline
\end{array}
\]
After checking the square integrability of these functions near the endpoints, we see that the endpoints have the following indices:
\[
\renewcommand{\arraystretch}{1.5}
\begin{array}{r|l|r}
	\hline
	\text{endpoint} &\text{parameters} 
	&                                                                 \text{index} \\
	\hline
	-\infty & & 0 \\
	+\infty & \Re(k)>0 & 0 \\
	+\infty & k=0, \beta\in\cc\setminus\rr_{>0}  & 0 \\
	+\infty & \Re(k)=0,k\neq0  &2 \\
	+\infty & k=0,\beta > 0  &2 \\
	\hline
\end{array}
\]

\vspace{1em}
Here is a description of the basic holomorphic family of Schr\"odinger operators with Morse potentials:
\begin{theoreme} \label{thm:morse}
	For $\beta,k\in \cc$ with $\Re(k)>0$  there exists a unique closed operator in the sense of $L^2(\mathbb{R})$, denoted $M_{\beta,k}$ which on $C_\mathrm{c}^\infty(\rr)$ it is given by \eqref{eq:morse}.
	It forms a holomorphic family of closed operators.
	
	Its  spectrum and point spectrum are
	\begin{align}
		\sigma(M_{\beta,k})&=	[0,\infty[\;\cup\;	\sigma_\mathrm{p}(M_{\beta,k}),\\
		\sigma_\mathrm{p}(M_{\beta,k})&=	\Big\{
			-m^2 \;\Big|\; m=\frac{\beta}{2k}-n-\frac{1}{2},\quad \Re(m)>0,\quad n \in\nn_0
		\Big\}
	.\label{specmorse}
	\end{align}
	Outside of the spectrum, its resolvent is given by
	\begin{equation}\label{eq:resol-Morse} 		
(M_{\beta,k}+m^{2})^{-1}(x,y)
		:= \;\Gamma\big(\tfrac{1}{2}+m-\tfrac{\beta}{2k}\big)
		\begin{cases}
			I_{\frac{\beta}{2k},m}(2k\e^{x}) K_{\frac{\beta}{2k},m}(2k\e^{y}), & \text{if }x<y,\\
			I_{\frac{\beta}{2k},m}(2k\e^{y}) K_{\frac{\beta}{2k},m}(2k\e^{x}), & \text{if }y<x.
		\end{cases}
	\end{equation}
\end{theoreme}

\begin{proof}
	The uniqueness of a closed realization of $N_{k,m}$ for $\Re(k)>0$ follows from the table above. $I_{\frac{\beta}{2k},m}(2k\e^{x})$ is then square integrable near $-\infty$ and $K_{\frac{\beta}{2k},m}(2k\e^{x})$ is square integrable near $+\infty$. We compute the Wronskian:
	\begin{equation}
		\cW(I_{\frac{\beta}{2k},m}(2k\e^{x}), K_{\frac{\beta}{2k},m}(2k\e^{x}))
		= \frac{1}{\Gamma\big(\tfrac{1}{2}+m-\tfrac{\beta}{2k}\big)}.
	\end{equation}
	
	Now, by \eqref{eq:k}  the kernel on the rhs of \eqref{eq:resol-Morse} is a candidate of the resolvent of $M_{\beta,k}$. 
	The boundedness is obtained  with the help of  Lemma \ref{lem:Morse}, similarly as in the previous (sub)section. 
\end{proof}

\begin{lemma}\label{lem:Morse}
	Let $k,\beta\in\cc$ with $\Re(k)>0$ and fix $m$ with $\Re(m)>0$ and $-m^2$ outside of \eqref{specmorse}.
	Let $R_{\beta,k}(-m^2)$ be the operator with the kernel  on the right hand side of \eqref{eq:resol-Morse}. 
	Then $R_{\beta,k}(-m^2)$ is bounded and depends analytically on $\beta,k$.
\end{lemma}

\begin{proof}
	The proof is very similar to that of Lemma \ref{lem:holomorphic-Rkm2}.
	We use the method of Section \ref{Boundedness of resolvent}. 
	We split the resolvent as 	in \eqref{split1} and \eqref{split2} with $c=0$. 
	Then we prove the boundedness of $R_{--}$ and $R_{++}$.
\end{proof}

\vspace{1em}
\begin{remark}
	Clearly, the family of Morse potentials extends the family of exponential potentials:
	\begin{equation}
		M_{0,k}=M_k ,
	\end{equation}
	i.e., if $\beta = 0$, then the Morse potential is the exponential potential covered in Section \ref{sec:exp}.
\end{remark}
\begin{remark}
	If $k=0$ and $\beta \in \cc \setminus \rr_+$, then after scaling, the Morse potential is the exponential potential covered in Section \ref{sec:exp}.
	This case is covered in Figure \ref{fig:exp_plot_c} with $\beta=-c$.
	If $k=0$ and $\beta>0$, then after scaling, Morse potentials is the negative exponential  covered in Section \ref{sec:negative-exp}.
	
	Let $A=\frac{xp+px}{2}$ where $p=-\i\partial_x$. 
	Let $U_\tau:=\e^{\i\tau A}$ be the {\em dilation operator}. It acts on functions as follows:
	\[
	(U_\tau f)(x) = \e^{\frac{\tau}{2}} f(\e^{\tau} x).
	\]
	Then with $U_\tau$, we have the following identity: 
	\begin{equation}
		U_{\ln 2}^{-1}\; M_{\beta,0}\; U_{\ln 2}
		= U_{\ln 2}^{-1} \Big( -\partial_x^2-\beta\e^x \Big) U_{\ln 2}
		= \frac14 \Big( -\partial_x^2 + (-4\beta) \e^{2x} \Big)
		= \frac14 M_{\sqrt{-4\beta}} \; .
	\end{equation}
\end{remark}

\begin{remark}
	We have not analyzed the case $\Re(k)=0$, except for $k=0$, which is discussed above. We leave it for future research.
\end{remark}

\subsection{Isotonic oscillator}\label{sec:isot}

The {\em isotonic harmonic oscillator} is formally defined by 
\begin{equation}\label{eq:isot}
	N_{k,m}:=-\partial_v^2+\big(m^2-\tfrac14\big)\frac{1}{v^2}+
	k^2 v^2. 
\end{equation} 
It appears in physics as the radial part of the radially symmetric harmonic oscillator in any dimension $>1$.
The name ``isotonic'' indicates that the frequencies in all directions are the same \cite{WeissmannJortner}.
We will interpret \eqref{eq:isot} as a closed operator on $L^2(\mathbb{R}_+)$.
Without restricting the generality, we can assume that $\Re(k)\geq0$.

Consider the change of variables $r = \frac{v^2}{2}$. We have the formal identity,
which connects the Whittaker operator with the isotonic oscillator:
\begin{equation} \label{eq:isot-guaging0}
	-\partial_r^2 
	+ \big(\tfrac{m^2}{4} - \tfrac{1}{4}\big) \frac{1}{r^2}
	-\frac{\beta}{r} 
	+ k^2
	=
	v^{-\frac{3}{2}}\Big(-\partial_v^2 + \big(m^2 - \tfrac{1}{4} \big)\frac{1}{v^2} +
	k^2 v^2 - 2\beta\Big)v^{-\frac{1}{2}}.
\end{equation} 
Therefore, eigenfunctions of the isotonic oscillator can be expressed in terms of the Whittaker functions. 
The details are described in Appendix \ref{Eigenequation of isotonic oscillator}, where the functions $\isoI$ and $\isoK$ are introduced. 
We will use them in the following table describing eigenvalues and corresponding eigenfunctions of \eqref{eq:isot}:
\[
\renewcommand{\arraystretch}{1.5}
\begin{array}{l|l|l}
	\hline
	\text{eigenvalue} &\text{parameters}& \text{eigenfunctions} \\
	\hline
	2\beta&\Re(k)\geq0,\quad  k\neq0 & \isoI_{\frac{\beta}{k},\pm m}(\sqrt{k} v),\quad \isoK_{\frac{\beta}{k},m}(\sqrt{k} v)\, \\
	-p^2 \text{ with } p \neq 0&k=0 & \cI_{\pm m}(pv), \quad \cK_m(pv) \\
	\phantom{-}0 &k=0,\quad m\neq0& v^{\frac{1}{2} \pm m}\\
	\phantom{-}0 &k=0,\quad m=0& v^{\frac12},\quad
	v^{\frac12}\ln v\\
	\hline
\end{array}
\]
After checking the square integrability of these functions near the endpoints, we see that the endpoints have the following indices:
\[
\renewcommand{\arraystretch}{1.5}
\begin{array}{r|l|r}
	\hline
	\text{endpoint} &\text{parameters} 
	& \text{index} \\
	\hline
	0 & |\Re(m)|<1            & 2 \\
	0 & |\Re(m)|\geq 1  & 0\\
	+\infty & &  0\\
	\hline
\end{array}
\]

\vspace{1em}
Let us describe the basic family of isotonic harmonic oscillators:
\begin{theoreme} \label{thm:isotonic} 
For $\Re(k)\geq0$ and $\Re(m)\geq1$  there exists a unique closed operator $N_{k,m}$ in the sense of $L^2(\rr_+)$ given on $C_\mathrm{c}^\infty(]0,\infty[)$ by \eqref{eq:isot}. 
It uniquely extends by analyticity in $m$ to $\Re(k)\geq0$ and $\Re(m)>-1$.
\begin{enumerate} 
	\item For $\Re(k)>0$, $\Re(m)>-1$ we have a holomorphic family with the spectrum
	\begin{alignat}{2}
		\sigma(N_{k,m})=\sigma_\mathrm{p}(N_{k,m}) &= \{ 2\beta=2k\big(2n+m+1\big)\;|\;
		n\in\nn_0\}.\quad &&
			\end{alignat}
	Outside of its spectrum, its resolvent is given by
	\begin{equation}
			(N_{k,m}-2\beta)^{-1}(u,v) 
			:=\;\tfrac{1}{2}\;\Gamma\Big(\tfrac12+\tfrac m2-\tfrac{\beta}{2k}\Big)
			\displaystyle
			\begin{cases}
				\isoI_{\frac{\beta}{k},m}(\sqrt{k} u)\;
				\isoK_{\frac{\beta}{k},m}(\sqrt{k} v)\,
				& \text{if }0<u<v,\\
				\isoI_{\frac{\beta}{k},m}(\sqrt{k} v)\;
				\isoK_{\frac{\beta}{k},m}(\sqrt{k} u)\,
				& \text{if }0<v<u.
			\end{cases}
			\label{eq:resol-isotonic}
	\end{equation}
		
	\item  Let $\Re(k)=0$, $k\neq0$ with $k=\ii\ell$, $\ell>0$. Then, we have $\sigma(N_{ k ,m})=\mathbb{R}$,  $\sigma_\mathrm{p}(N_{ k ,m})=\emptyset$ and, for $\pm\Im2\beta>0$
	\begin{equation}
		(N_{k,m}-2\beta)^{-1}(u,v) 
		:=\;\tfrac{1}{2}\;\Gamma\Big(\tfrac12+\tfrac m2-\tfrac{\beta}{\mp2 k }\Big)
		\displaystyle
		\begin{cases}
			\isoI_{\frac{\beta}{\mp k },m}(\sqrt{\mp k } u)\;
			\isoK_{\frac{\beta}{\mp k },m}(\sqrt{\mp k } v)\,
			& \text{if }0<u<v,\\
			\isoI_{\frac{\beta}{\mp k },m}(\sqrt{\mp k } v)\;
			\isoK_{\frac{\beta}{\mp k },m}(\sqrt{\mp k } u)\,
			& \text{if }0<v<u.
		\end{cases}
		\label{eq:resol-isotonic-}
	\end{equation}
		For $\Re(k)=0$, $k\neq0$ with $k=-\ii\ell$, $\ell>0$, we set  $N_{k,m}=N_{- k ,m}$.
	\item  The case $k=0$ coincides with the   Bessel operators: $N_{0,m} = H_m$. 
\end{enumerate}
\end{theoreme}

\begin{proof}
For $\Re(k)\geq0$ and $\Re(m)>-1$, we define $N_{k,m}$ by setting the b.c.\ at zero with $v^{\frac12+m}$. 
We check that
\begin{equation}
	\cW\big(v^{\frac12+m},\isoI_{\frac{\beta}{k},m}(\sqrt{k} v)\big)=0
\end{equation}
and $\isoI_{\frac{\beta}{k},m}(\sqrt{k} v)$ is square integrable near $0$.

Let us now consider Cases 1,2,3 separately.

For $\Re(k)>0$, $\isoK_{\frac{\beta}{k},m}(\sqrt{k} v)$ is square integrable near $+\infty$.
We check that
\begin{equation}
	\cW\big(\isoI_{\frac{\beta}{k},m}(\sqrt{k} v), \isoK_{\frac{\beta}{k},m}(\sqrt{k} v)\big)
	= \frac{2}{\Gamma\big(\frac12+\frac m2-\frac{\beta}{2k}\big)}.
\end{equation}
Now, \eqref{eq:k} yields the kernel on the rhs of \eqref{eq:resol-isotonic} as a candidate of the resolvent of $N_{k,m}$.
In Lemma \ref{lemma}.1., we check that it is bounded and depends analytically on parameters $k,m$. This proves Case 1.

For $k=\ii \ell$, $\ell>0$, and $\pm\Im(2\beta)>0$, we have 
\begin{equation}
	\Re\Big(\frac\beta{\mp k }\Big)=\mp\frac{\Im\beta}{\ell}<0,
\end{equation} 
thus, $\isoK_{\frac{\beta}{\mp k }}(\sqrt{\mp k } v)$, which can be estimated by $Cv^{-\frac12+\Re\frac{\beta}{\mp k}}$, is square integrable near $+\infty$. 
Therefore, using \eqref{eq:k}  the rhs of \eqref{eq:resol-isotonic-} is a candidate of the kernel of the resolvent of  $N_{ k ,m}$.
We check that $\pm\Im\beta<0$ and $\Re(m)>-1$, $\ell>0$ implies that there are no solutions of
\begin{equation}
	2\beta=2(\mp\ii\ell)(2n+m+1).
\end{equation}
Therefore, besides $\rr$ there is no spectrum of $N_{ k ,m}$.
In Lemma \ref{lemma}.2, we check that \eqref{eq:resol-isotonic-} is bounded and depends analytically on $m$. This proves Case 2.

Finally, Case 3 was treated before.
\end{proof}

\vspace{1em}

\begin{lemma} 
For $k\neq0$, let $R_{k,m}(2\beta)$ denote the operator with the integral kernel \eqref{eq:resol-isotonic}. Let  $\Re(m)>-1$.
\begin{enumerate}
	\item If $\Re(k)>0$ and $2\beta$ is outside of the spectrum, then $R_{k,m}(2\beta)$ is bounded and depends analytically on $k,m$.
	\item If $k=\mp\ii \ell$, $\ell>0$, and $\pm\Im(2\beta)>0$, then $R_{k,m}(2\beta)$ is bounded and depends analytically on $m$.
\end{enumerate}\label{lemma}
\end{lemma}

\begin{proof}
We use Lemma \ref{lemma-bound} with, say, $c=\frac{1}{\sqrt{|k|}}$, to split $R_{k,m}$.
We have the estimates
\begin{align}\label{esti2}
	|R_{k,m}^{--}(-2\beta;u,v)|		
	& \quad\leq
	\begin{cases} 
		Cu^{\frac12+\Re(m)}v^{\frac12-\Re(m)} & \text{if }0<u<v<c,\\
		Cv^{\frac12+\Re(m)}u^{\frac12-\Re(m)} & \text{if }0<v<u<c,
	\end{cases}\quad m\neq0;\\
	\label{esti3}
	|R_{k,0}^{--}(-2\beta;u,v)|		
	& \quad\leq
	\begin{cases}
		Cu^{\frac12}v^{\frac12}|\ln(v)| & \text{if }0<u<v<c,\\
		Cv^{\frac12}u^{\frac12}|\ln(u)| & \text{if }0<v<u<c,
	\end{cases}\\
	\label{esti1}
	|R_{k,m}^{++}(-2\beta;u,v)|		
	& \quad\leq
	\begin{cases}
		Cu^{-\frac12-\Re\frac\beta{k}}v^{-\frac12+\Re\frac\beta{k}}\e^{(u^2-v^2)\Re(k)} & \text{if }c<u<v,\\
		Cv^{-\frac12-\Re\frac\beta{k}}u^{-\frac12+\Re\frac\beta{k}}\e^{(v^2-u^2)\Re(k)} & \text{if }c<v<u.
	\end{cases}
\end{align}
Now, we use \eqref{esti2}, \eqref{esti3} and the Hilbert-Schmidt estimate to prove the boundedness of $R_{k,m}^{--}(2\beta)$, both in Case 1 and 2.
Then, we use \eqref{esti1} to prove the boundedness of $R_{k,m}^{++}(2\beta)$. 
We treat separately for Case 1 and Case 2.

Let $\Re(k)>0$. Set $b:=-\Re(\tfrac\beta{k})$. Using by  \eqref{esti1} we see that that
\begin{equation}\label{esto1}
	R_{k,m}^{++}(-2\beta;u,v)
	\leq C_1\e^{-|u-v|(\frac{(\ln u-\ln v)}{(u-v)}b+(u+v)\Re(k))}.
\end{equation}
The function 
\[
[c, \infty[ \times [c, \infty[ \;\ni (u, v) \mapsto \frac{\ln u - \ln v}{u - v}b + (u + v)\Re(k)
\]
is continuous and goes to \( +\infty \) as \( u \to +\infty \) or \( v \to +\infty \). 
Therefore, away from a bounded set, it can be bounded from below by some \( \epsilon > 0 \). 
Hence, \eqref{esto1} can be estimated from above by \( \leq C_1 \e^{-|u - v| \epsilon} \). 
Now \eqref{esto1} is the kernel of a bounded operator by Young's inequality or Schur's criterion. 
This proves Case 1.

Now assume \( \Re(k) = 0 \). Let \( b := -\Re\left(\frac{\beta}{k}\right) \), as before. We now have \( b > 0 \).
We have
\begin{align}
	|R_{k,m}^{++}(-2\beta;u,v)|		
	& \quad\leq
	\begin{cases}
		Cu^{-\frac12+b}v^{-\frac12-b} & \text{if }c<u<v,\\
		Cv^{-\frac12+b}u^{-\frac12-b} & \text{if }c<v<u.
	\end{cases}
\end{align}

Without loss of generality, one may let $c=0$.	
Making substitutions $u=\e^t$, $v=\e^s$, we obtain
\begin{align}
	(f|R_{k,m}^{++}(-2\beta;u,v)g)
	&=\int_0^\infty \int_u^\infty|f(u)||g(v)|u^{-\frac12+b}v^{-\frac12-b}+(u\leftrightarrow v) \,\d v \d u \\
	\notag
	&=\int_{-\infty}^{\infty} \int_{t}^{\infty} |f(\e^t)||g(\e^s)|\e^{\frac{t}{2}+\frac{s}{2}-b|t-s|}+(s \leftrightarrow t)\, \d s \d t \\
	\notag
	&\leq\frac{2}{b}\Big(\int |f(\e^t)|^2\e^t\,\d t\Big)^{\frac12} \Big(\int |g(\e^s)|^2\e^s\,\d s\Big)^{\frac12}\\
	&=\frac2b \Big(\int |f(u)|^2\,\d u\Big)^{\frac12} \Big(\int |g(v)|^2\,\d v\Big)^{\frac12}.
\end{align}
This shows Case 2.
\end{proof}

\vspace{2em}

\begin{remark} 
The operators $N_{k,\mp\frac12}$ are the harmonic oscillators on $\mathbb{R}_+$ with the Neumann and Dirichlet conditions, respectively.
We can denote them as follows:
\begin{equation} \label{remarkdn}
	N_{k,-\frac12}=:N_k^\mathrm{N},
	\quad
	N_{k,\frac12}=:N_k^\mathrm{D}.
\end{equation}
Here are their spectra:
\begin{align}\label{eq:NDNN}
	\sigma(N_{k}^\mathrm{N}) &= \{k\big(4n+1\big)\;|\; n\in\nn_0\},\\
	\sigma(N_{k}^\mathrm{D}) &= \{k\big(4n+3\big)\;|\; n\in\nn_0\}.
\end{align}
In the next section we will consider them again.\label{remarkdn1}
\end{remark}

\begin{remark}For $-1<\Re(m)<1$, we can also consider mixed 
b.c.. They can be analyzed similarly as for Bessel operators.
\end{remark}

\section{Harmonic Oscillator}\label{sec:harmonic}

This section is devoted to the \textit{harmonic oscillator}, formally defined on $L^2(\rr)$ by
\begin{equation}\label{eq:harmonic}
	N_k := -\partial_v^2 + k^2 v^2.
\end{equation}
We will consider its closed realizations on $L^2(\rr)$.
Without loss of generality, we assume that $\Re(k) \geq 0$.

For real $k$, the self-adjoint realization of $N_k$ is one of the best-known operators in Quantum
Mechanics. For complex $k$, its closed realization is an interesting example of an
operator with sometimes surprising properties, and has been studied, e.g., in \cite{Davis,Pravda-Starov}.

To describe the eigenfunctions of \eqref{eq:isot}, we will use Weber functions, defined in Appendix \ref{sec:weber}.
\[
\renewcommand{\arraystretch}{1.5}
\begin{array}{l|l|l}
	\hline
	\text{eigenvalue} &\text{parameters}& \text{eigenfunctions} \\
	\hline
	2\beta &\Re(k)\geq0,\quad k\neq0 &
	\isoI_{\frac{\beta}{k},\pm}(\sqrt{k} v), \quad\isoK_{\frac{\beta}{k}}(\pm\sqrt{k} v) \\
	-p^2\text{ with }p\neq0& k = 0  & \e^{pv},\quad \e^{-pv}\\
	\phantom{-}0 & k = 0  & v,\quad 1\\
	\hline
\end{array}
\]
The endpoints have always index $0$:
\[
\renewcommand{\arraystretch}{1.5}
\begin{array}{r|l|r}
	\hline
	\text{endpoint} &\text{parameters}
	& \text{index} \\
	\hline
	-\infty & \Re(k)\geq0 & 0 \\
	+\infty & \Re(k)\geq0  & 0\\
	\hline
\end{array}
\]

\vspace{1em}
For the harmonic oscillator, we  have a unique closed realization for all parameters:
\begin{theoreme}\label{thm:harmonic}
Let $\Re(k)\geq0$. There exist a unique closed operator in the sense of $L^2(\rr),$ which on
$C_\mathrm{c}^\infty(\mathbb{R)}$ coincides with \eqref{eq:harmonic}.
It will be denoted  $N_k$.
\begin{enumerate}
	\item $\{\Re(k)>0\}\ni k \mapsto N_k$ is a holomorphic family of closed operators with  the following spectrum:
	\begin{align}
		\sigma(N_k ) =		\sigma_\mathrm{p}(N_k)&=\{k(2n+1)\ |\ n\in\mathbb{N}_0\}.
	\end{align}
	For  $2\beta$ away of its spectrum, its resolvent is given by
	\begin{equation}\label{eq:resol-harmonic1}
	(N_k - 2\beta)^{-1}(u,v) = 
		\tfrac{1}{2\sqrt{2\pi}}\,2^{\frac{\beta}{k}}\,\Gamma(\tfrac{1}{2}-\tfrac{\beta}{k})
		\displaystyle
		\begin{cases}
			\isoK_{\frac{\beta}{k}}(-\sqrt{k} u)\;
			\isoK_{\frac{\beta}{k}}(\sqrt{k} v)\,
			& \text{if } u<v,\\
			\isoK_{\frac{\beta}{k}}(-\sqrt{k} v)\;
			\isoK_{\frac{\beta}{k}}(\sqrt{k} u)\,
			& \text{if } v<u.
		\end{cases}
	\end{equation}
	
	\item Let $\Re(k)=0$, $k\neq0$ with $k=\ii\ell$, $\ell>0$. Then, we have $\sigma(N_{ k })=\mathbb{R}$,  $\sigma_\mathrm{p}(N_0)=\emptyset$  and, for $\pm\Im2\beta>0$,
	\begin{equation}\label{eq:resol-harmonic2}
		(N_{ k } - 2\beta)^{-1}(u,v) = 
		\tfrac{1}{2\sqrt{2\pi}}\,2^{\frac{\beta}{\mp k }}\,\Gamma(\tfrac{1}{2}-\tfrac{\beta}{\mp\i\ell})
		\displaystyle
		\begin{cases}
			\isoK_{\frac{\beta}{\mp k }}(-\sqrt{\mp k } u)\;
			\isoK_{\frac{\beta}{\mp k }}(\sqrt{\mp k } v)\,
			& \text{if } u<v,\\
			\isoK_{\frac{\beta}{\mp k }}(-\sqrt{\mp k } v)\;
			\isoK_{\frac{\beta}{\mp k }}(\sqrt{\mp k } u)\,
			& \text{if } v<u.
		\end{cases}
	\end{equation}
		For $\Re(k)=0$, $k\neq0$ with $k=-\ii\ell$, $\ell>0$, we set  $N_{ k }=N_{- k }$.
              \item Finally, $\sigma( N_0)=[0,\infty[\,,$  $\sigma_\mathrm{p}(N_0)=\emptyset$ and, for $\Re(p)>0$,
	\begin{equation}
		(N_0+p^2)^{-1}(u,v)=\frac{\e^{-p|u-v|}}{2p}.
	\end{equation}
\end{enumerate}
\end{theoreme}

\begin{proof} 
The uniqueness of a closed realization follows immediately from the table.

For $\Re(k)>0$, $\isoK_{\frac\beta{k}}(-\sqrt{k}v)$ is square integrable near $-\infty$ and $\isoK_{\frac\beta{k}}(\sqrt{k}v)$ is square integrable near $-\infty$. 
Using \eqref{conne} and \eqref{wron}, and then Legendre's duplication formula, we get the Wronskian:
\begin{equation}
	\cW(\isoK_{\frac{\beta}{k}}(-\sqrt{k}u),\isoK_{\frac{\beta}{k}}(\sqrt{k}u))
	=\frac{4\pi}{\Gamma(\tfrac{1}{4}-\tfrac{\beta}{2k})\,\Gamma(\tfrac{3}{4}-\tfrac{\beta}{2k})}
	=\frac{2\sqrt{2\pi}}{2^{\tfrac{\beta}{k}}\,\Gamma(\tfrac{1}{2}-\tfrac{\beta}{k})}.
\end{equation}
Now, we apply \eqref{eq:k} to have the kernel on the rhs of \eqref{eq:resol-harmonic1} as a candidate of the resolvent of $N_{k}$, denoting it by $R_{k}(-2\beta)$.
The boundedness of $R_{k}(-2\beta)$ is immediate from the subsection on the isotonic oscillator.

For $k=\ii \ell$, $\ell>0$, and $\pm\Im(2\beta)>0$, we have
\begin{equation}
	\Re\Big(\frac\beta{\mp\ii\ell}\Big)=\mp\frac{\Im\beta}{\ell}<0,
\end{equation}
thus, $\isoK_{\frac{\beta}{\mp\ii\ell}}(\sqrt{\mp\ii\ell} v)$ is square integrable near $+\infty$ and $\isoK_{\frac{\beta}{\mp\ii\ell}}(-\sqrt{\mp\ii\ell} v)$ near $-\infty$.
The boundedness of $R_{k}(-2\beta)$ is  proven as for the isotonic oscillator.

The case $k=0$ is just the well-known free 1d Laplacian.
\end{proof}

As is well-known, the propagator for $N_k$ can be expressed in elementary functions. It is given by the so-called {\em Mehler's formula}
\begin{equation}
	\e^{-\i t \frac12 N_1} (u,v) = \frac{1}{\sqrt{\pi}}\sqrt{\frac{\rho}{1-\rho^2}} \exp\left(-\frac{(1+\rho^2)(u^2+v^2)-2\rho uv}{2(1-\rho^2)}\right),
\end{equation}
where $\rho = \e^{-\i t}$.

\begin{remark}
Recall that in \eqref{remarkdn} we introduced the harmonic oscillators with the Neumann and Dirichlet boundary conditions $N_k^{\mathrm{N}/\mathrm{D}}=N_{k,\mp\frac12}$ on $\rr_+$, as special cases of the isotonic harmonic oscillator.
They are closely related to the harmonic oscillator $N_k$ on $\rr$.

Let $L_\pm^2(\rr)$ denote the subspace of $L^2(\rr)$ consisting of even, resp. odd functions. 
Let us define the unitary operators
\begin{align}
	U_\pm:L^2(\rr_+) &\to L_\pm^2(\rr),\\
	(U_+ g)(v) &:=\frac{1}{\sqrt2}g(|v|),\quad v\in\rr;\\
	(U_- g)(v) &:=\frac{\sgn(v)}{\sqrt2}g(|v|),\quad v\in\rr;\\
	\intertext{so that }  
	(U_\pm^*f)(v) &=\sqrt2f(v),\quad v\in\rr_+.
\end{align}
Then
\begin{equation}
	N_k=U_+N_k^\mathrm{N} U_+^*+U_-N_k^\mathrm{D} U_-^*,
\end{equation}
where ``$+$'' can be replaced by ``$\oplus$'' in the sense of the direct sum $L^2(\rr)=L_+^2(\rr)\oplus L_-^2(\rr)$. 
Consequently,
\begin{equation}
	\sigma(N_k) = \sigma(N_k^\rD) \cup \sigma(N_k^\rN)
\end{equation}
where the spectra of $N_k^\rD$ and $N_k^\rN$ were computed in \eqref{eq:NDNN}.
\end{remark}

\vspace{2em}

\appendix

\section{Bessel equation}
\label{Bessel equation}

There are several kinds of Bessel equations, all equivalent to one another. Their main application is the Helmholtz equation in $d$ dimensions:
\begin{equation}\label{helmo}
	(-\Delta_d + E)f = 0.
\end{equation}
If $E > 0$, \eqref{helmo} can be simplified to $E = 1$, a case often referred to as \emph{hyperbolic}. Conversely, if $E < 0$, which can similarly be simplified to $E = -1$, the case is sometimes known as \emph{trigonometric}. The radial part of \eqref{helmo} on spherical harmonics of order $\ell$ is
\begin{equation}\label{helmo1}
	\Big(-\partial_r^2 - \frac{(d-1)}{r}\partial_r
	+ \Big(\big(\ell + \tfrac{d}{2} - 1\big)^2 - \big(\tfrac{d}{2} - 1\big)^2\Big)\frac{1}{r^2} + E\Big)f = 0.
\end{equation}
The versions of \eqref{helmo1} for various dimensions are equivalent by gauging (or conjugating) the operator with a power of $r$.\footnote{By gauging (or conjugating) the operator $A$ with a function $f(r)$, we mean replacing it with $f(r)Af(r)^{-1}$. See, for example, \eqref{gas}, \eqref{eq:0F1-Hyperbolic2dBessel}, \eqref{eq:Hyperbolic1dBesselequation}, \eqref{eq:1F1}, or \eqref{eq:1dWhittaker-gauage} for gauging with a power of $r$, and \eqref{gas} and \eqref{eq:1F1} for gauging with an exponential.}

The standard Bessel equation corresponds to the 2d trigonometric case, while the so-called modified Bessel equation corresponds to the 2d hyperbolic case. However, sometimes it is convenient to use versions of the Bessel equation for other dimensions. In our paper, in some cases 1d Bessel functions are more convenient; in others, 2d Bessel functions. Therefore, we will discuss both.

All forms of the Bessel equation are equivalent to the so-called ${}_0F_1$ equation, which is not as well-known. In fact, one could argue that the ${}_0F_1$ equation and its standard solutions ${\bf F}_\alpha$ and $U_\alpha$ have a simpler theory than the usual Bessel equation and functions.

In this section, we first discuss the ${}_0F_1$ equation and its solutions. Then we describe hyperbolic and trigonometric 1d and 2d Bessel equations and functions.

In the whole appendix, the variables $w$, $z$, $v$, and $r$ are complex, although elsewhere in this paper we usually restrict them to $]0,\infty[$.

\subsection{${}_0F_1$ equation}

Let $c\in\cc$. The standard solution of the {\em the ${}_0F_1$ equation }
\begin{equation}\label{equa}
	(w\,\p_w^2+c\,\p_w-1)f(w)=0
\end{equation}
is the {\em hypergeometric ${}_0F_1$ function}
\begin{equation}
	{}_0F_1(c;w):=\sum_{n=0}^\infty\frac{w^n}{(c)_nn!}
      \end{equation}
    where for $k\in\nn_0$
\[
	(c)_k = {\Gamma(c+k)\over\Gamma(c)} = 
	\begin{cases} 
		c(c+1)(c+2)\cdots(c+k-1),	& \mbox{if $k\ge 1$\,;}	\\ 
		1, 							& \mbox{if $k=0$}\,.	\\ 
	\end{cases}
\]
If $c\neq0,-1,-2,\dots$, then it is the only solution of the ${}_0F_1$ equation $\sim1$ at 0.
It is convenient to normalize it differently:
\begin{equation}
	{}_0\mathbf{F}_1(c;w):=\frac{{}_0F_1(c;w)}{\Gamma(c)}=\sum_{n=0}^\infty\frac{w^n}{n!\Gamma(c+n)}, 
\end{equation}
so that it is defined for all $c$.

The ${}_0F_1$ equation can be reduced to a special class of the confluent equation by the so-called {\em Kummer's 2nd transformation}:
\begin{equation}\label{gas}
	z\,\p_z^2+c\,\p_z-1
	=\frac{4}{w}\e^{-\frac{w}{2}}\Big(z\,\partial_w^2+(2c-1-w)\,\partial_w-c+\frac12\Big)\e^{\frac{w}{2}},
\end{equation}
where $w=\pm 4\sqrt{z}$, $z=\frac{1}{16}w^2$.
Using this, we can derive an expression for the ${}_0F_1$ function in terms of the confluent function \eqref{eq:kummer-function}:
\begin{align*}
	{}_0F_1(c;z)&=\e^{\mp2\sqrt{z}}{}_1F_1\Big(\frac{2c-1}{2}; 2c-1;\pm4\sqrt{z}\Big). 
\end{align*}

Instead of $c$, it is often more natural to use $\alpha := c - 1$, and rewrite \eqref{equa} as
\begin{equation}\label{equa1}
	(z\,\p_z^2+(\alpha +1)\,\p_z-1)v(z)=0, 
\end{equation}
and set
\begin{equation}
	F_\alpha (z):={}_0F_1(\alpha+1;z),\quad {\bf F}  _\alpha (z):= {}_0{\bf F}_1 (\alpha+1;z) .
\end{equation}

The following function is also a solution of the ${}_0F_1$ equation \eqref{equa1}:
\begin{align*}
	U_\alpha (z)&:=\e^{-2\sqrt z} z^{-\frac{\alpha}{2} -\frac14}
	{}_2 F_0\Big(\frac12+\alpha,\frac12-\alpha;-;-\frac{1}{4\sqrt z}\Big),
\end{align*}
where we used the ${}_2F_0$ function \eqref{2f0}.
Obviously,
\[
	U_\alpha (z)=z^{-\alpha }U_{-\alpha }(z).
\]

As $|z|\to\infty$ and $|\arg z|<\frac{\pi}{2}-\varepsilon$, we have 
\begin{equation}\label{saddle1}
	U_\alpha (z)\sim{\rm exp}(- 2z^{\12}) z^{-\frac{\alpha }2-\frac14}.
\end{equation}
$U_\alpha$  is a unique solution of  \eqref{equa1}  with this property.

We can express $U_\alpha$ in terms of the solutions of with a simple behavior at zero 
\begin{align}\label{conni}
	U_\alpha (z)
	&=\frac{\sqrt\pi}{\sin\pi (-\alpha )} {\bf F}  _\alpha (z)
	+\frac{\sqrt \pi}{\sin\pi \alpha }
	z^{-\alpha } {\bf F}  _{-\alpha }(z).
\end{align}

\subsection{Hyperbolic 2d Bessel equation}
\label{Hyperbolic 2d Bessel equation}

The usual {\em modified Bessel} equation has the form
\begin{equation}\label{eq:whi00+}
	\Big ( -\partial_{r}^2 -\frac1r\partial_r+ \frac{m^2}{r^2}  +1\Big)g=0. 
\end{equation}
We  use the name the {\em hyperbolic 2d Bessel equation} for \eqref{eq:whi00+}.
It is equivalent to the ${}_0F_1$ equation:
\begin{equation}\label{eq:0F1-Hyperbolic2dBessel}
	w^{\frac{ m }{2}} \big(w\,\p_w^2+(1+m)\,\p_w-1\big) w^{-\frac{m}{2}}
	= \partial_r^2+\frac{1}{r}\partial_r-1-\frac{m^2}{r^2},
\end{equation}
where $w=\frac{r^2}{4}$, $r=\pm 2\sqrt w$.

The {\em hyperbolic 2d Bessel function} $I_{m}$ is defined by
\begin{equation}\label{eq:relationship cI and I}
	I_m(r):=\Big(\frac{r}2\Big)^{m}{}_0\mathbf{F}_1\Big(m+1;\frac{r^2}{4}\Big).
\end{equation}
We have the Wronskian
\begin{equation} 
	\cW(I_m,I_{-m})=-\frac{2\sin(\pi m)}{\pi r},
\end{equation}
and for $m\in\mathbb{Z}$
\begin{equation}\label{refle}
	I_m(r)=I_{-m}(r). 
\end{equation}

The {\em 2d Macdonald function} $K_m$ is defined by 
\begin{align}
	K_{-m}(r)=  K_m(r)
	&:=\frac{\sqrt\pi}{2}\Big(\frac{r}{2}\Big)^{m}U_m\Big(\frac{r^2}{4}\Big) \\
	&=\frac{\pi}{2\sin(\pi m)}\big(I_{-m}(r)-I_{m}(r)\big). 
\end{align}

\subsection{Hyperbolic 1d Bessel equation}
\label{Hyperbolic 1d Bessel equation}

The {\em hyperbolic 1d Bessel equation}
\begin{equation}\label{eq:whi0}
	\Big ( -\partial_{r}^2 + \Big ( m^2 - \frac14 \Big ) \frac{1}{r^2}  +1\Big)g=0. 
\end{equation}
is  equivalent to the hyperbolic 2d Bessel equation by a simple gauge
transformation:
\begin{equation}\label{eq:Hyperbolic1dBesselequation}
	r^{-\frac12}	\Big ( -\partial_{r}^2 + \Big ( m^2 - \frac14 \Big ) \frac{1}{r^2} + 1 \Big) r^{\frac12}
	= -\partial_{r}^2 -\frac1r\partial_r+ \frac{m^2}{r^2}  + 1 .
\end{equation}
The {\em hyperbolic 1d Bessel function} $\mathcal{I}_{m}$ is defined by
\begin{equation}\label{eq:relationship cI and I2}
	\cI_m(r):=\sqrt\pi\Big(\frac{r}2\Big)^{\frac12+m}{}_0\mathbf{F}_1\Big(m+1;\frac{r^2}{4}\Big)
	=\sqrt{\frac{\pi r}{2}} I_m(r).
\end{equation}
We have the Wronskian: 
\begin{equation}
	\cW(\cI_m,\cI_{-m})=-\sin(\pi m).
\end{equation}
and for $m\in\mathbb{Z}$, 
\begin{equation}\label{refle2}
	\cI_m(r)=\cI_{-m}(r).
\end{equation}

The {\em 1d Macdonald function} $\cK_m$ is defined by 
\begin{align}
	\cK_{-m}(r):=  \cK_m(r)
	&=\Big(\frac{r}{2}\Big)^{\frac12+m}U_m\Big(\frac{r^2}{4}\Big) 
	=\sqrt{\frac{2 r}{\pi}} K_m(r)\label{eq:relationshop cK and K}\\
	&=\frac{1}{\sin(\pi m)}\big(\cI_{-m}(r)-\cI_{m}(r)\big). 
\end{align}

\subsection{Trigonometric 2d Bessel equation}\label{sec:Standard (trigonometric) Bessel equation}

The usual {\em Bessel equation}, which can be called the {\em trigonometric 2d Bessel equation}, has the form
\begin{equation}\label{eq:whi00}
	\Big(-\partial_{r}^2 - \frac{1}{r}\partial_r + \frac{m^2}{r^2} - 1\Big)g = 0. 
\end{equation}
We can pass from the hyperbolic 2d to the trigonometric 2d Bessel equation by the substitution $r \to \i r$.

The (usual) {\em Bessel function} (or the {\em trigonometric 2d Bessel function}) is
\begin{equation}\label{eq:J and I}
	J_m(r) = \e^{\pm\i\frac{\pi}{2}m} I_m(\e^{\mp\i\frac{\pi}{2}} r).   
\end{equation}
We also have two {\em Hankel functions}:
\begin{equation}\label{eq:H and K}
	H_m^\pm(r) = \frac{2}{\pi} \e^{\mp\i\frac{\pi}{2}(m+1)} K_m(\e^{\mp\i\frac{\pi}{2}} r).
\end{equation}

Note that the traditional notation for  $H_m^\pm$ is $H_m^{(1)}$ and $H_m^{(2)}$. The authors believe that their notation is more handy.

\subsection{Trigonometric 1d Bessel equation}
\label{Trigonometric 1d Bessel equation}

We also have the trigonometric 1d Bessel equation:
\begin{equation}\label{eq:whi0-2}
	\Big ( -\partial_{r}^2 + \Big ( m^2 - \frac14 \Big ) \frac{1}{r^2}  -1\Big)g=0. 
\end{equation}
We can pass from  hyperbolic 1d to trigonometric Bessel equations by the substitution $r\to \i r$.

We can introduce various  kinds of solutions of the 1d  trigonometric Bessel equation: the \emph{1d Bessel function} 
\begin{equation}\label{eq: cJ_m(r) and cI_m}
	\cJ_m(r)
	:= \e^{\pm\i\frac{\pi}{2}(m+\frac12)}\cI_m(\e^{\mp\i\frac{\pi}{2}}r)
	=\sqrt{\frac{\pi r}{2}} J_m(r),
\end{equation}
and the \emph{1d Hankel function of the 1st/2nd kind}
\begin{equation}\label{eq: cH_m(r) and cK_m}
	\cH_m^\pm(r)
	:= \e^{\mp\i\frac{\pi}{2}(m+\frac12)} \cK_m(\e^{\mp\i\frac{\pi}{2}} r)
	=\sqrt{\frac{\pi r}{2}} H^\pm_m(r).
\end{equation}

\section{Whittaker equation}
\label{Whittaker equation}

The ${}_1F_1$ equation, the ${}_2F_0$ equation, and the Whittaker equation are equivalent to one another by certain substitutions and gauge transformations. 
In this section, we briefly describe conventions and properties of solutions to these equations.

Note that our definitions of Whittaker functions differ slightly from some of the literature, e.g., \cite{NIST}. 
In particular, we use what is sometimes called {\em Olver's normalization}, which is advantageous because it avoids singularities for the parameters under consideration. We follow the conventions of \cite{Derezinski2014hyper,DR}.

The Whittaker equation is equivalent to the radial part of the Schr\"odinger equation with the Coulomb potential in any dimension:
\begin{equation}
	\Big(-\Delta_d - \frac{\beta}{r} + \frac{1}{4}\Big)f = 0.
\end{equation}
Therefore, there exists a variant of the Whittaker equation for any dimension. The standard one corresponds to $d = 1$. 
We will also find it convenient to consider the Whittaker equation for $d = 2$.

In the following subsections, we review several equations, equivalent to one another: the ${}_1F_1$ equation, the ${}_2F_0$ equation, the 1d Whittaker equation, the 2d Whittaker equation, and the eigenequation of the isotonic oscillator.

\subsection{${}_1F_1$ equation}

The {\em ${}_1F_1$ hypergeometric equation}, also called the {\em confluent equation} has the form
\begin{equation}\label{eq:confluent-equation}
	\big(r\,\p_r^2+(c-r)\,\p_r-a\big)f(r)=0.
\end{equation}
Its standard solution is {\em Kummer's confluent hypergeometric function} ${}_1F_1(a;c;\cdot)$ 
\begin{equation}\label{eq:kummer-function}
	{}_1F_1(a;c;r): = \suma{k=0}{\infty}\frac{(a)_k}{(c)_k}\frac{r^k}{k!}.
\end{equation}
It is the only solution of \eqref{eq:confluent-equation} behaving as $1$ in the vicinity of $r=0$. It is often convenient to normalize it differently:
\begin{equation}\label{eq:whi}
	{}_1\mathbf{F}_1(a;c;r): = \suma{k=0}{\infty}\frac{(a)_k}{\Gamma(c+k)}\frac{r^k}{k!} = \frac{{}_1F_1(a;c;r)}{\Gamma(c)}.
\end{equation}
It satisfies 1st Kummer's identity
\begin{equation}
	{}_1F_1(a;\,c;\,r) = \e^r {}_1F_1(c-a;\,c;\,-r). 
\end{equation}

\subsection{${}_2F_0$ equation}
The {\em ${}_2F_0$ hypergeometric equation}  has the form
\begin{equation}
	\big(w^2\partial_w^2+(-1+(1+a+b)w)\,\partial_w + ab\big)v(w)=0.
\end{equation}
The ${}_2F_0$ equation has a distinguished solution, which can be expressed as a limit of the Gauss hypergeometric function:
\begin{equation}\label{2f0}
	{}_2F_0(a,b;-;w):=\lim_{c\to\infty}{}_2F_1(a,b;c;cw),
\end{equation}
where we take the limit over $|\arg(c)-\pi|<\pi-\varepsilon$ with $\varepsilon>0$, and the above definition is valid for $w\in\cc \backslash[0,+\infty[$.
Obviously one has
\begin{equation}\label{eq:obvio}
	{}_2 F_0(a,b;-;w)={}_2F_0(b,a;-;w).
\end{equation}
The function extends to an analytic function on the universal cover of $\cc\backslash\{0\}$ with a branch point of an infinite order at 0, and the following asymptotic expansion holds:
\begin{equation*}
	{}_2F_0(a,b;-;w)\sim\sum_{n=0}^\infty\frac{(a)_n(b)_n}{n!}w^n,\quad |\arg(w)|<\pi-\varepsilon.
\end{equation*}

\subsection{The 1d Whittaker equation}
\label{The 1d Whittaker equation}

The usual \emph{Whittaker equation} corresponds to dimension 1 and has the following form:
\begin{equation}\label{eq:Whittaker-hyper}
	\Big(-\partial_r^2 +\big(m^2 - \frac{1}{4}\big)\frac{1}{r^2} - \frac{\beta}{r}+\frac{1}{4}\Big)v(r) = 0.
\end{equation}
It can be reduced to the ${}_1F_1$-equation,
\begin{equation}\label{eq:1F1}
	-r^{\frac{1}{2}\mp m}\e^{\frac{r}{2}}\Big(-\partial_r^2 + \big(m^2 - \frac{1}{4}\big)\frac{1}{r^2} - \frac{\beta}{r} + \frac{1}{4}\Big)r^{\frac{1}{2}\pm m}\e^{-\frac{r}{2}}
	\;=\; r\partial_r^2 + (c-r)\partial_r - a
\end{equation}
for the parameters $c = 1 \pm 2m$ and $a = \frac{1}{2} \pm m-\beta$.
Here the sign $\pm$ has to be understood as two possible choices.
The following function solves the  Whittaker equation \eqref{eq:Whittaker-hyper}:
\begin{equation}\label{eq:serie_I}
	\cI_{\beta,m}(r) 
	:= r^{\frac{1}{2}+m}\e^{\mp \frac{r}{2}} {}_1\mathbf{F}_1\Big(\frac{1}{2}+m\mp\beta;\,1+2m;\,\pm r\Big).
\end{equation}
Note that the sign independence comes from the $1^{\mathrm{st}}$ Kummer's identity. We have the Wronskian
\begin{equation}
	\cW(\cI_{\beta,m},\cI_{\beta,-m})=-\frac{\sin(2\pi m)}{\pi}.
\end{equation}

The Whittaker equation is also equivalent to the ${}_2F_0$ equation.
Indeed by setting $w=-r^{-1}$ we obtain
\begin{align*}
	&-r^{2-\beta}
	\e^{\frac{r}{2}}\Big(-\partial_r^2 + \big(m^2 - \frac{1}{4}\big)\frac{1}{r^2} - \frac{\beta}{r} + \frac{1}{4}\Big)r^{\beta}\e^{-\frac{r}{2}}\\
	&= w^2\partial_w^2+(-1+(1+a+b)w)\,\partial_w + ab.
\end{align*}
for the parameters $a = \frac{1}{2} +m -\beta$ and $b = \frac{1}{2} -m -\beta$.

We  define
\begin{equation*}
	\cK_{\beta,m}(r) 
	:=r^\beta\e^{-\frac{r}{2}} {}_2F_0\Big(\frac12+m-\beta,\frac12-m-\beta;-;-r^{-1}\Big),
\end{equation*}
which is thus a solution of the  Whittaker equation \eqref{eq:Whittaker-hyper}.
The symmetry relation \eqref{eq:obvio} implies that
\begin{equation}\label{eq_sym}
	\cK_{\beta,m}(r) = \cK_{\beta,-m}(r).
\end{equation}
The following connection formulas hold for $2m\notin\zz$\;\!:
\begin{align}\label{connec}
	\cK_{\beta,m}(r)& = \frac{\pi}{\sin(2\pi m)}\Big(-\frac{\cI_{\beta,m}(r)}{\Gamma(\frac{1}{2}-m-\beta)}
	+ \frac{\cI_{\beta,-m}(r)}{\Gamma(\frac{1}{2}+m-\beta)}\Big),\\
	\cI_{\beta,m}(r)& = \frac{\Gamma(\frac{1}{2}-m+\beta)}{2\pi}\Big(\e^{\i\pi m}\cK_{-\beta,m}(\e^{\i\pi}r) + \e^{-\i\pi m}\cK_{-\beta,m}(\e^{-\i\pi}r)\Big).\label{conni2}
\end{align}

Here is an estimate for small $r$:
\begin{align}
	\cI_{\beta,m}(r)=\frac{r^{\frac12+m}}{\Gamma(1+2m)}(1+O(r)).
\end{align}
If $\Re(m)>0$, then
\begin{equation}
	\cK_{\beta,m}(r) = \frac{\Gamma(2m)}{\Gamma(\frac12 + m - \beta)}\Big( r^{\frac12-m} + o(|r|^{\frac12 - m})\Big).
\end{equation}

For large $r$, if $\varepsilon>0$, then for $|\arg r|<\pi-\varepsilon$,
\begin{align}
	\cK_{\beta,m}(r)=r^\beta\e^{-\frac{r}{2}}\big(1+O(r^{-1})\big).
\end{align}
This together with \eqref{conni2} implies the estimates
\begin{align}
	\cI_{\beta,m}(r)&=\frac{1}{\Gamma(\frac12 + m - \beta)}
	r^{-\beta} \e^{\frac{r}{2}} (1+ O(r^{-1})),\quad |\arg r| <
	\frac{\pi}{2}-\epsilon;\\
	|	\cI_{\beta,m}(r)|&\leq\frac{1}{|\Gamma(\frac12 + m - \beta)|}
	r^{-\Re(\beta)}  (1+ O(r^{-1})),\quad |\arg r|=\frac{\pi}{2}.
\end{align}

The relation of the functions $\cI_{0,m}$ and  $\cK_{0,m}$ with $\cI_m$ and $\cK_m$ reads
\begin{equation}\label{eq:K_{0,m}}
	\cI_{0,m}(r)=\frac{2}{\Gamma(\frac12+m)}\cI_m\Big(\frac{r}{2}\Big),\quad
	\cK_{0,m}(r) 
	=\cK_m\Big(\frac{r}{2}\Big).
\end{equation}

\subsection{The 2d Whittaker equation}\label{sec:newWhittaker}

The 2d Whittaker equation has the form
\begin{equation}\label{eq:Whittaker-hyper2}
	\Big(-\partial_r^2 -\frac1r\partial_r+\frac{m^2}{r^2} - \frac{\beta}{r}+\frac{1}{4}\Big)v = 0.
\end{equation}
It is equivalent to the 1d Whittaker equation by a simple gauge transformation
\begin{equation}\label{eq:1dWhittaker-gauage}
	r^{-\frac12}	\Big ( -\partial_{r}^2 + \Big ( m^2 - \frac14 \Big ) \frac{1}{r^2} - \frac{\beta}{r}+\frac{1}{4}\Big) r^{\frac12}
	= -\partial_{r}^2 -\frac1r\partial_r+ \frac{m^2}{r^2}  - \frac{\beta}{r}+\frac{1}{4}.
\end{equation}

In order to make our presentation more transparent, similarly as in the case of Bessel functions,  beside the usual (1d) Whittaker functions,
we define \emph{2d Whittaker functions}, which solve the 2d Whittaker equation:
\begin{align}
	I_{\beta,m}(r)&:=\sqrt{\frac{2}{\pi r}}\, \cI_{\beta,m}(r)
	\qquad	K_{\beta,m}(r):=\sqrt{\frac{\pi}{2r}}\, \cK_{\beta,m}(r).
\end{align}

\subsection{Eigenequation of isotonic oscillator}
\label{Eigenequation of isotonic oscillator}

The eigenequation of the isotonic harmonic oscillator \eqref{eq:isot} (with $k=1$) has the form
\begin{equation}\label{isoto}
	\Big(-\partial_v^2+\big(m^2-\tfrac14\big)\frac{1}{v^2}+v^2-2\beta\Big)f(v) = 0. 
\end{equation}

Let us recall the following identity \eqref{eq:isot-guaging0} involving the change of variables $r=\frac{v^2}{2}$:
\begin{equation} \label{eq:isot-guaging}
	-\partial_r^2 
	+ \big(\tfrac{ m^2}{4} - \tfrac{1}{4}\big) \frac{1}{r^2}
	-\frac{\beta}{r} 
	+ k^2
	=
	v^{-\frac32}\Big(-\partial_v^2 +\Big(m^2-  \tfrac{1}{4} \Big)\frac{1}{v^2}+
	k^2 v^2-2\beta\Big)v^{-\frac12}.
\end{equation} 
$\cI_{\frac\beta{2k},\frac{m}{2}}(2r)$, $\cK_{\frac\beta{2k},\frac{m}{2}}(2r)$ are annihilated by the left hand side of \eqref{eq:isot-guaging}.
Therefore,  \eqref{isoto} is  solved in terms of the following functions:
\begin{align}\label{iso1}
	\isoI_{\beta,m}(v)
	&:=v^{\frac12+m}\e^{\mp\frac{1}{2}v^{2}}{}_{1}\mathbf{F}_{1}\Big(\frac{1+m\mp\beta}{2};\,1+m;\,\pm
	v^{2}\Big)\\\label{iso2}
	&=v^{-\frac12}\cI_{\frac\beta2,\frac{m}2}(v^2) ,
	\\
	\isoK_{\beta,m}(v)
	&:=v^{\beta-\frac12}\e^{-\frac{1}{2}v^{2}}{}_{2}F_{0}\Big(\frac{1+m-\beta}{2},\frac{1-m-\beta}{2};-;-v^{-2}\Big)\\
	&=v^{-\frac12}\cK_{\frac\beta2,\frac{m}2}(v^2) 
	.
\end{align}
Note from \eqref{eq_sym} that
$\isoK_{\beta,m}(r)=\isoK_{\beta,-m}(r)$,
\begin{align}
	\cW(  \isoI_{\beta,m}, \isoI_{\beta,-m})&=-\frac{2\sin\pi m}{\pi},\\
	\label{connec1}
	\isoK_{\beta,m}(v)& = \frac{\pi}{\sin(\pi m)}\Big(-\frac{\isoI_{\beta,m}(v)}{\Gamma(\frac{1-m-\beta}{2})}
	+ \frac{\isoI_{\beta,-m}(v)}{\Gamma(\frac{1+m-\beta}{2})}\Big),\\
	\isoI_{\beta,m}(v)& = \frac{\Gamma(\frac{1-m+\beta}{2})}{2\pi}\Big(\e^{\i\frac{\pi}{2} m}\isoK_{-\beta,m}(\e^{\i\frac{\pi}{2}}v) + \e^{-\i\frac{\pi}{2} m}\isoK_{-\beta,m}(\e^{-\i\frac{\pi}{2}}v)\Big).\label{poroq}
\end{align}

For small $v$, we have
\begin{align}\label{poro}  
	\isoI_{\beta,m}(v)&=\frac{v^{\frac12+m}}{\Gamma(1+m)}(1+O(v^2)).
\end{align}
Hence,
\begin{equation}
	\isoK_{\beta,m}(v) = \frac{\Gamma(m)}{\Gamma(\frac{1 + m - \beta}{2})}\Big( v^{\frac12-m} + o(|v|^{\frac12 - m})\Big).
\end{equation}

For large $v$ we have the asymptotics
\begin{align}\label{poro2}
	\isoK_{\beta,m}(v)
	&=v^{\beta-\frac12}\e^{-\frac{v^2}{2}}(1+O(v^{-2})),\quad
	|\arg v|<\frac\pi2-\epsilon,\quad\epsilon >0.
\end{align}
Together with \eqref{poroq}, this implies
\begin{align}
	\isoI_{\beta,m}(v)
	&=\frac{1}{\Gamma(\frac{1 + m - \beta}2)}
	v^{-\beta-\frac12} \e^{\frac{v^2}{2}} (1+ O(v^{-2})),
	\quad |\arg v| < \frac{\pi}{4}-\epsilon;\\
	|\isoI_{\beta,m}(v)|
	&\leq\frac{1}{|\Gamma(\frac{1 + m - \beta}{2})|}
	v^{-\Re(\beta)-\frac12}  (1+ O(v^{-2})),
	\quad |\arg v|=\frac{\pi}{4}.
\end{align}

\section{Weber equation}\label{sec:weber}

The {\em Weber equation} (also called the {\em parabolic cylinder equation}) has the form
\begin{equation}\label{harmo}
	\left(-\partial_v^2 + v^2 - 2\beta\right)f(v) = 0.
\end{equation}
It is the eigenequation of the harmonic oscillator and a special case of \eqref{isoto}, the eigenequation of the isotonic oscillator, with $m = \pm\frac{1}{2}$.

Note that \eqref{harmo}, unlike \eqref{isoto}, does not have a singularity at $v = 0$. Therefore, its solutions are analytic at $v = 0$.

Let us introduce notation for distinguished solutions of \eqref{harmo}:
\begin{align}
	\label{harmo1}	\isoI_{\beta,\pm}(v) &:= \isoI_{\beta,\mp\frac{1}{2}}(v),\\
	\label{harmo2}	\isoK_{\beta}(v) &:= \isoK_{\beta,\pm\frac{1}{2}}(v).
\end{align}
(In the literature, they are called {\em Weber(-Hermite) functions} or {\em parabolic cylinder functions}).
Note that \eqref{harmo1} and \eqref{harmo2} should be understood as follows: we first define the functions on $[0,\infty[$ as in \eqref{iso1} and \eqref{iso2}; then we extend them analytically to the whole complex plane.

The equation \eqref{harmo} is invariant with respect to the mirror symmetry. Therefore, it is spanned by even and odd solutions.
It is easy to see that $\isoI_{\beta,+}$ is even and $\isoI_{\beta,-}$ is odd:
\begin{equation} 
	\isoI_{\beta,\pm}(-v) = \pm \isoI_{\beta,\pm}(v).
\end{equation}

The function $\isoK_\beta(v)$ has the decaying asymptotics \eqref{poro2} in the positive direction. 
In the negative direction, it usually blows up. $\isoK_\beta(-v)$ is also a solution of the Weber equation. 
It decays in the negative direction.

We have
\begin{align}\label{conne}
	\isoK_\beta(\pm v) &= \pm\frac{\isoI_{\beta,-}(v)}{\Gamma(\frac{3}{4} - \frac{\beta}{2})} -
	\frac{\isoI_{\beta,+}(v)}{\Gamma(\frac{1}{4} - \frac{\beta}{2})},\\
	\cW(\isoI_{\beta,+}, \isoI_{\beta,-}) &= \frac{2}{\pi}.\label{wron}
\end{align}

\vspace{2em}
{\bf Acknowledgement.}
J.D. was supported by National Science Center (Poland) under the grant UMO-2019/35/B/ST1/01651.
J. L. was supported by
the Swiss National Science Foundation through the NCCR SwissMAP,
the SNSF Eccellenza project PCEFP2\_181153, 
by the Swiss State Secretariat for Research and Innovation through the project P.530.1016 (AEQUA), and
Basic Science Research Program through the National Research Foundation of Korea (NRF) funded by the Ministry of Education (RS-2024-00411072).

\end{document}